\RequirePackage{xcolor}
\documentclass[peerreview,onecolumn,print,a4paper]{ieeecolor}     

\def\QED{\QEDopen}
\usepackage{arxiv}
\usepackage{cite}
\usepackage{amsmath,amssymb,amsfonts}
\usepackage{algorithmic}
\usepackage{graphicx}
\usepackage{textcomp}
	\usepackage{bbm}
\usepackage{mathtools}
\usepackage{color,soul}
\usepackage{url}
\usepackage{outlines}
\usepackage{comment}
\usepackage{enumerate}
\usepackage{ntheorem}
\usepackage{textgreek}
\usepackage[margin=2.5cm]{geometry}

\usepackage{tikz}
\usepackage{pgfplots}
\usepackage{grffile}
\pgfplotsset{compat=newest}
\usetikzlibrary{automata}
\usetikzlibrary{arrows, arrows.meta, decorations.markings, patterns}
\usepgfplotslibrary{patchplots, groupplots}
\usetikzlibrary{shapes, positioning, fit, backgrounds, spy}
\usepgfplotslibrary{fillbetween}

\usetikzlibrary{overlay-beamer-styles}
\usetikzlibrary{calc}

\usepackage{calc}

\definecolor{tudcyan}{RGB}{0,166,214}
\definecolor{tudmagenta}{RGB}{109,23,127}
\definecolor{tudpurple}{RGB}{29,28,115}
\definecolor{tudgraygreen}{RGB}{107,134,137}

\colorlet{lighttudcyan}{tudcyan!20}
\colorlet{lighttudmagenta}{tudmagenta!20}

\newlength{\hatchspread}
\newlength{\hatchthickness}
\newlength{\hatchshift}
\newcommand{\hatchcolor}{}
\tikzset{hatchspread/.code={\setlength{\hatchspread}{#1}},
	hatchthickness/.code={\setlength{\hatchthickness}{#1}},
	hatchshift/.code={\setlength{\hatchshift}{#1}},
	hatchcolor/.code={\renewcommand{\hatchcolor}{#1}}}
\tikzset{hatchspread=7pt,
	hatchthickness=0.5pt,
	hatchshift=0pt,
	hatchcolor=black}
\pgfdeclarepatternformonly[\hatchspread,\hatchthickness,\hatchshift,\hatchcolor]
{custom north west lines}
{\pgfqpoint{\dimexpr-2\hatchthickness}{\dimexpr-2\hatchthickness}}
{\pgfqpoint{\dimexpr\hatchspread+2\hatchthickness}{\dimexpr\hatchspread+2\hatchthickness}}
{\pgfqpoint{\dimexpr\hatchspread}{\dimexpr\hatchspread}}
{
	\pgfsetlinewidth{\hatchthickness}
	\pgfpathmoveto{\pgfqpoint{0pt}{\dimexpr\hatchspread+\hatchshift}}
	\pgfpathlineto{\pgfqpoint{\dimexpr\hatchspread+0.15pt+\hatchshift}{-0.15pt}}
	\ifdim \hatchshift > 0pt
	\pgfpathmoveto{\pgfqpoint{0pt}{\hatchshift}}
	\pgfpathlineto{\pgfqpoint{\dimexpr0.15pt+\hatchshift}{-0.15pt}}
	\fi
	\pgfsetstrokecolor{\hatchcolor}
	\pgfusepath{stroke}
}

\def\centerarc[#1](#2)(#3:#4:#5)
{ \draw[#1] ($(#2)+({#5*cos(#3)},{#5*sin(#3)})$) arc (#3:#4:#5); }

\usepackage{etoolbox}
\makeatletter
\@ifundefined{color@begingroup}%
{\let\color@begingroup\relax
	\let\color@endgroup\relax}{}%
\def\fix@ieeecolor@hbox#1{%
	\hbox{\color@begingroup#1\color@endgroup}}
\patchcmd\@makecaption{\hbox}{\fix@ieeecolor@hbox}{}{\FAILED}
\patchcmd\@makecaption{\hbox}{\fix@ieeecolor@hbox}{}{\FAILED}
\let\old@makecaption=\@makecaption
\usepackage{subcaption}
\let\@makecaption=\old@makecaption

\def\BibTeX{{\rm B\kern-.05em{\sc i\kern-.025em b}\kern-.08em
		T\kern-.1667em\lower.7ex\hbox{E}\kern-.125emX}}

\newtheorem{defn}{Definition}
\newtheorem{rem}{Remark}
\newtheorem{prop}{Proposition}
\newtheorem{lem}{Lemma}
\newtheorem{thm}{Theorem}
\newtheorem{cor}{Corollary}

\newtheorem{example}{Example}

\newcommand*{\tran}{^{\mkern-1.5mu\mathsf{T}}\!}  

\def\d{\ensuremath{\mathrm{d}}}
\def\imag{\ensuremath{\mathrm{i}}}
\def\Id{\ensuremath{\mathrm{Id}}}
\DeclareMathOperator{\sign}{sign}

\DeclareMathOperator{\Post}{Post}

\DeclareMathOperator{\InfLimInf}{ILI}
\DeclareMathOperator{\SupLimSup}{SLS}
\DeclareMathOperator{\InfLimAvg}{ILA}

\def\Rob{\mathrm{Rob}}

\DeclareMathOperator{\bigO}{\mathcal{O}}
\def\norm[#1]{\left|#1\right|}
\def\shortnorm[#1]{|#1|}
\def\bisim{\ensuremath{\cong}}
\DeclareMathOperator{\cl}{cl}

\DeclareMathOperator{\ang}{angle}

\def\Xs{\mathcal{X}}
\def\Ys{\mathcal{Y}}
\def\Us{\mathcal{U}}
\def\Ws{\mathcal{W}}
\def\Vs{\mathcal{V}}

\def\No{\mathbb{N}_{0}}
\def\N{\mathbb{N}}
\def\R{\mathbb{R}}
\def\Q{\mathbb{Q}}
\def\S{\mathbb{S}}
\def\P{\mathbb{P}}
\def\C{\mathbb{C}}

\def\Rs{\mathcal{R}}
\def\Os{\mathcal{O}}

\def\Ss{\mathcal{S}}
\def\Ks{\mathcal{K}}
\def\Qs{\mathcal{Q}}

\def\As{\mathcal{A}}  
\def\Es{\mathcal{E}}  
\def\Cs{\mathcal{C}}  
\def\Bs{\mathcal{B}}  
\def\Us{\mathcal{U}}  
\def\Ps{\mathcal{P}}  

\def\Ts{\mathcal{T}}

\def\nx{{n_{\mathrm{x}}}}

\def\xv{\boldsymbol{x}}
\def\yv{\boldsymbol{y}}

\def\vv{\boldsymbol{v}}
\def\xiv{\boldsymbol{\xi}}

\def\dv{\boldsymbol{d}}

\def\lv{\boldsymbol{l}}

\def\zv{\boldsymbol{z}}
\def\pv{\boldsymbol{p}}

\def\Am{\boldsymbol{A}}
\def\Bm{\boldsymbol{B}}

\def\Em{\boldsymbol{E}}
\def\I{\mathbf{I}}
\def\O{\mathbf{0}}

\def\Mm{\boldsymbol{M}}
\def\Nm{\boldsymbol{N}}

\def\Pm{\boldsymbol{P}}
\def\Qm{\boldsymbol{Q}}

\def\Tm{\boldsymbol{T}}

\def\Sm{\boldsymbol{S}}

\def\Wm{\boldsymbol{W}}
\def\Vm{\boldsymbol{V}}

\def\Km{\boldsymbol{K}}

\def\Om{\boldsymbol{O}}

\def\e{\mathrm{e}}
\def\piconst{\textup{\textpi}}
\def\uppi{{\pi}}

\def\myvdots{\vbox{\baselineskip4\p@ \lineskiplimit\z@
		\hbox{.}\hbox{.}\hbox{.}}}

\def\sincos{\begin{bmatrix}\sin\theta & \cos\theta\end{bmatrix}\tran}

\def\customproof[#1]{\noindent\hspace{2em}{\itshape {#1}:}}

\begin{document}
	
	\title{Chaos and Order in Event-Triggered Control}
	
	\author{Gabriel de Albuquerque Gleizer and Manuel Mazo Jr. 
		\thanks{This work is supported by the European Research Council through the SENTIENT project (ERC-2017-STG \#755953).}
		\thanks{G.~de A.~Gleizer and M.~Mazo Jr. are with the Delft Center for Systems and Control,
			Delft Technical University, 2628 CD Delft, The Netherlands
			{\tt\small \{g.gleizer, m.mazo\}@tudelft.nl}}%
	}
	\maketitle
	
	\begin{abstract}
		
		Event-triggered control (ETC) is claimed to provide enormous reductions in sampling frequency when compared to periodic sampling, but little is formally known about its generated traffic. This work shows that ETC can exhibit very complex, even chaotic traffic, especially when the triggering condition is aggressive in reducing communications. First, we characterize limit traffic patterns by observing invariant lines and planes through the origin, as well as their attractivity. Then, we present abstraction-based methods to compute limit metrics, such as limit average and limit inferior inter-sample time (IST) of periodic ETC (PETC), with considerations to the robustness of such metrics, as well as measuring the emergence of chaos. The methodology and tools allow us to find ETC examples that provably outperform periodic sampling in terms of average IST. In particular for PETC, we prove that this requires aperiodic or chaotic traffic.
		
	\end{abstract}

	\section{INTRODUCTION}
	Since the seminal paper from \cite{tabuada2007event},  
	event-triggered control (ETC) has been considered a disruptive method for sampled-data control implementations over digital media. The astonishingly simple design and stability analysis methods proposed by Tabuada cast new light on the idea of aperiodic sampling, which had been studied since the 1950s \cite{hufnagel1958aperiodic} and gained renewed interest in the early 2000s \cite{aastrom2002comparison}. The idea behind ETC is natural: instead of sampling periodically, sample only when ``needed'' based on some significant event; therefore, massive reductions in communications, as well as in energy of battery-powered motes, can be achieved, enabling new control applications with cheap hardware, or larger networks of control systems. Unsurprisingly, immense interest followed, and a lot of effort was dedicated in the following years to design better event-triggering mechanisms \cite{wang2008event, girard2015dynamic} for e.g., perturbed systems, extend applications to output-feedback control \cite{heemels2012introduction}, or make implementations more practical, as the periodic ETC (PETC) from \cite{heemels2013periodic}, where events are checked periodically.
	
	While understanding of ETC's control performance and stability has reached a high level of maturity thanks not only to the aforementioned papers and their successors, but also to the hybrid-systems formalism of \cite{goebel2012hybrid}, the comprehension of ETC's sampling patterns and performance is severely lacking. In other words, how much communication savings can ETC achieve when compared to, e.g., periodic sampling? What is an ETC's average sampling time, or average \emph{inter-event time}? Most analyses in ETC papers are only concerned with a lower bound estimate of the \emph{minimum inter-event time} (MIET), to prove the absence of Zeno behavior. We argue that, while this is absolutely important to obtain, it is often conservative and does not prove that the ETC's performance is better than a well-designed periodic sampling time. Not surprisingly, most, if not all, ETC papers contain numerical simulations showing inter-event time trajectories and statistics of average inter-event time, to give evidence of ETC's practical relevance. 
	
	Only recently, effort has been conducted to model ETC's generated traffic, which we split here into two categories.\footnote{It is also worth mentioning the approach of \cite{linsenmayer2018performance}, which proposed an event-triggering mechanism that ensures given traffic criteria in terms of a token bucket model. Despite very interesting, we veer away from this approach because it is unclear whether adding conditions to enforce traffic patterns could actually degrade the sampling performance of the original mechanism.}  The first category is aimed at qualitative understanding of asymptotic properties of inter-event times of ETC \cite{postoyan2019interevent, rajan2020analysis, postoyan2022explaining}; in these papers, the studies are dedicated to two-dimensional linear time-invariant (LTI) systems, and some conditions are given to show when traffic converges to periodic sampling or oscillatory patterns. In particular, \cite{postoyan2022explaining} allows to approximately compute average inter-sample for such planar systems when the triggering parameters are sufficiently small. The second category aims at taming the highly variable inter-event times of ETC for scheduling purposes, and for that it relies on traffic models using finite-state abstractions under the framework of \cite{tabuada2009verification}: such models have been developed for continuous ETC for LTI systems in \cite{kolarijani2016formal, mazo2018abstracted}, for PETC in \cite{fu2018traffic, gleizer2020scalable}, and nonlinear systems in \cite{delimpaltadakis2020homogeneous, delimpaltadakis2020traffic}, while only in \cite{gleizer2020towards} longer-term traffic predictions are addressed. Based on this second category, we have recently developed tools to compute the smallest (across initial states) average inter-sample time (SAIST) of an LTI system under PETC \cite{gleizer2021hscc, gleizer2021computing}, by using weighted automata \cite{chatterjee2010quantitative} as abstractions.
	
	There are issues involved in both the qualitative and quantitative analyses in the present literature. On the quantitative side, the obtained metrics lack a sense of robustness: that is, a given PETC system may have a SAIST of, e.g., 1 time unit, but this may be only observed by a neglibible, 0-measure subset of initial conditions. If all other initial states converge to some other traffic pattern with higher SAIST, e.g., 3, this much higher value is clearly a more representative performance metric. However, ETC systems may not have stable traffic patterns, which hints on a problem of the qualitative side of the literature: the focus has been given to relatively simple, stable traffic patterns. Perhaps due to experiments with small triggering parameters, more complex traffic patterns have not been described or observed. We show in this work that ETC systems can exhibit chaotic traffic, and as such a stable representative traffic pattern may not be found. The emergence of chaotic traffic also forces us to carefully define robust metrics for ETC before attempting to compute them.
	
	The present work makes an attempt to expand the qualitative understanding of ETC's asymptotic traffic patterns and bridge it to the quantitative approach of \cite{gleizer2021hscc, gleizer2021computing}, focusing on LTI systems and a common class of quadratic triggering conditions \cite{heemels2013periodic}. For that, we first characterize limit metrics of interest, such as limit inferior and limit average, and observe that they are related to the asymptotic properties of the traffic. This is the starting point for our main contributions: (i) presenting limit behaviors of LTI ETC systems and methods to compute them, not limited to $\R^2$; (ii) classifying limit behaviors in terms of stable vs.~unstable, periodic vs.~aperiodic, orderly vs.~chaotic; and (iii) based on this classification, expanding the results from \cite{gleizer2021hscc, gleizer2021computing} for PETC to compute robust metrics. 
	We propose auxiliary concepts and obtain results that may be useful on their own right: (i) we show that the law of evolution of state samples can be regarded as a map of the projective space to itself, which allows us to conclude that stationary traffic patterns are always exhibited in CETC for odd-dimensional systems (Theorem~\ref{thm:topofixed}); (ii) we show that if a PETC that renders the origin globally asymptotically stable (GES) converges to a periodic traffic pattern, then this traffic pattern can be used as a (multi-rate) periodic sampling schedule (Prop.~\ref{prop:petchurwitz}) --- this does not necessarily happen to CETC; (iii) we provide a stability characterization for outputs of a system, when these outputs come from a finite set; and (iv) we present the notion of behavioral entropy (Def.~\ref{def:entropy}) as a measure of chaos of a system's set of output trajectories, how to compute this quantity in an abstraction (Theorem \ref{thm:valueofentropy}), and show that this quantity is an upper bound of the concrete system's (Prop.~\ref{prop:entropybounds}).
	
	This paper is organized as follows: §\ref{sec:prob} presents the basic ETC formulation, how the inter-event times can be computed, and the main problem statement. The qualitative side of the work, presenting limiting behaviors and their general properties is given in §\ref{sec:behaviors}, where we are able to establish conditions for which periodic patterns occur and the associated states that generate them. In doing that, we explore their local attractivity and the emergence of chaotic invariant sets. This paves the way for the quantitative side of this work in §\ref{sec:symbolic} using symbolic abstractions, where we properly define robust limit metrics for PETC taking chaos into consideration, provide methods to estimate PETC's behavioral entropy using the abstraction, establish when traffic patterns are not involved in chaos, and describe how to estimate the desired robust limit metrics. Numerical examples are given in §\ref{sec:numerical}, and a discussion and conclusions in §\ref{sec:conclusions}.
	
	\section{Mathematical Preliminaries}
	
	We denote by $\No$ the set of natural numbers including zero, $\N \coloneqq \No \setminus \{0\}$, $\N_{\leq n} \coloneqq \{1,2,...,n\}$, and $\R_+$ the set of non-negative reals.
	We denote by $|\xv|$ the norm of a vector $\xv \in \R^n$, but if $s$ is a sequence or set, $|s|$ denotes its length or cardinality, respectively. For a square matrix $\Am \in \R^{n \times n},$ $\lambda(\Am) \subset \C^n$ is the set of its eigenvalues, and $\lambda_i(\Am)$ is the $i$-th largest-in-magnitude. The complex conjugate of $z \in \C^n$ is denoted by $z^*$ . The set $\S^n$ denotes the set of symmetric matrices in $\R^n$. For $\Pm \in \S^n$, we write $\Pm \succ \O$ ($\Pm \succeq \O$) if $\Pm$ is positive definite (semi-definite); $\lambda_{\max}(\Pm)$ ($\lambda_{\min}(\Pm)$) denotes its maximum (minimum) eigenvalue. For a set $\Xs\subseteq\Omega$, we denote by $\cl(\Xs)$ its closure, $\partial\Xs$ its boundary, and $\bar{\Xs}$ its complement: $\Omega \setminus \Xs$. 
	We often use a string notation for sequences, e.g., $\sigma = abc$ reads $\sigma(1) = a, \sigma(2) = b, \sigma(3) = c.$ Powers and concatenations work as expected, e.g., $\sigma^2 = \sigma\sigma = abcabc.$ In particular, $\sigma^\omega$ denotes the infinite repetition of $\sigma$. An infinite sequence of numbers is denoted by $\{a_i\} \coloneqq a_0, a_1, a_2, ... .$ %
	We apply a function $f : \Xs \to \Ys$ to a set $\As \subseteq \Xs$ the usual way, $f(\As) \coloneqq \{f(\xv)  \mid \xv \in \As\}.$ %
	For a relation $\Rs \subseteq \Xs_a \times \Xs_b$, its inverse is denoted as $\Rs^{-1} = \{(x_b, x_a) \in \Xs_b \times \Xs_a \mid (x_a, x_b) \in \Rs\}$. 
	
	We say that an autonomous system $\dot{\xiv}(t) = f(\xiv(t))$ is globally exponentially stable (GES) if there exist $M < \infty$ and $b > 0$ such that every solution of the system satisfies $|\xiv(t)| \leq M\e^{-bt}|\xiv(0)|$ for every initial state $\xiv(0)$. 
	When needed to avoid ambiguity, we use $\xiv_{\xv}(t)$ to denote a trajectory from initial state $\xiv(0) = \xv.$
	
	\subsection{Chaos}
	
	Consider the map $f:\Xs \to \Xs$ and the discrete-time system (or recursion) $\xv(k+1) = f(\xv(k))$. A set $\Ys \subset \Xs$ is said to \emph{fixed} or \emph{invariant} if $f(\Ys) = \Ys$, \emph{forward invariant} if $f(\Ys) \subseteq \Ys$, and \emph{periodic} if there is some $m \in \N$ such that $f^m(\Ys) = \Ys$. %
	The \emph{forward orbit} of a point $\xv$ is $\Os(\xv) \coloneqq \{f^k(\xv) \mid k \in \N_0\}.$ Obviously, every forward orbit is forward invariant. 
	Whilst there are multiple slightly different definitions of chaos, we use the concept of \cite{robinson1999dynamical}, which relies on the notions of \emph{transitivity} and \emph{sensitivity to initial conditions.}

	\begin{defn}[Transitivity{{\cite[Sec.~2.5]{robinson1999dynamical}}}]\label{def:trans}
		A map $f:\Xs\to\Xs$ is said to be \emph{(topologically) transitive} on an invariant set $\Ys$ if the forward orbit of some point $p \in \Xs$ is dense in $\Ys$. From the Birkhoff Transitivity Theorem, this is equivalent to the following property: for every two open subsets $\Us$ and $\Vs$ of $\Ys$, there is a positive integer $n$ such that $f^n(\Us) \cap \Vs \neq \emptyset.$
	\end{defn}
	
	If $f$ is transitive, points starting arbitrarily close to each other can drift away but will come arbitrarily close back to each other after enough iterations.
	
	\begin{defn}[Sensitivity to initial conditions{{\cite[Sec.~3.5]{robinson1999dynamical}}}]
		A map $f:\Xs\to\Xs$, $\Xs$ being a metric space, is said to be \emph{sensitive to initial conditions} on an invariant set $\Ys \subseteq \Xs$ if there is an $r>0$ such that, for each point $\xv\in\Ys$ and for each $\epsilon>0$, there exists a point $\yv\in\Ys$ satisfying $d(\xv,\yv)<\epsilon$ and a $k \in \N$ with $d(f^k(\xv),f^k(\yv)) \geq r$.
	\end{defn}
	
	\begin{defn}[Chaos{{\cite[Sec.~3.5]{robinson1999dynamical}}}]
		A map $f:\Xs\to\Xs$, $\Xs$ being a metric space, is said to be \emph{chaotic on an invariant set} $\Ys$ provided (i) $f$ is transitive on $\Ys$, and (ii) $f$ is sensitive to initial conditions on $\Ys$.
	\end{defn}

	In case a chaotic system is additionally \emph{ergodic}\footnote{See \cite{robinson1999dynamical} for a rigorous definition of ergodicity. We skip the definition and present a simplified version of the Birkhoff Ergodic Theorem due to readability and space considerations.} the celebrated Birkhoff Ergodic Theorem is particularly useful when one is interested in limit average metrics:
	
	\begin{thm}[Birkhoff Ergodic Theorem \cite{robinson1999dynamical}]\label{thm:birk}
		Assume $f : \Xs \to \Xs$ is an ergodic function with ergodic measure $\mu$, and let $g : \Xs \to \R $ be a $\mu$-integrable function. Then, 
		$$\lim_{n\to\infty}\frac{1}{n+1} \sum_{i=1}^n g \circ f^i(x) = \int_\Xs g(x) \d\mu(x)$$ for $\mu$-almost every $x$.
	\end{thm}
	
	As a consequence, if $f$ is ergodic, the time-average converges to the same value from almost every initial condition.

\subsection{Invariants of linear systems}\label{ssec:technicallemmas}

Most of the analysis of limit behaviors of (P)ETC on linear systems involve studying invariants of linear systems, their stability and relationship with quadratic cones. In this subsection we provide some definitions and results on this topic.

\begin{defn}[Homogeneous set] A set $\Qs \subseteq \R^n$ is called \emph{homogeneous} if $\xv \in \Qs \implies \lambda\xv \in \Qs, \forall \lambda \in \R \setminus \{0\}$. \end{defn}

For the next results, we borrow a few definitions from previous work \cite{gleizer2021computing} concerning square matrices.

\begin{defn}[Mixed matrix \cite{gleizer2021computing}] Consider a matrix $\Mm \in \R^{n \times n}$ and let $\lambda_i, i \in \N_{\leq n}$ be its eigenvalues sorted such that $\norm[\lambda_i] \geq \norm[\lambda_{i+1}]$ for all $i$. We say that $\Mm$ is \emph{mixed} if, for all $i < n$, $\norm[\lambda_i] = \norm[\lambda_{i+1}]$ implies that $\Im(\lambda_i) \neq 0$ and $\lambda_i = \lambda_{i+1}^*$.
\end{defn}

Mixed matrices are diagonalizable and do not have distinct eigenvalues of the same magnitude, with the exception of pairs of complex conjugate eigenvalues.

\begin{defn}[Matrix of irrational rotations \cite{gleizer2021computing}] A matrix $\Mm \in \R^{n \times n}$ is said to be of \emph{irrational rotations} if the arguments of all of its complex eigenvalues are irrational multiples of $\piconst$.
\end{defn}

The set of mixed matrices of irrational rotations is of full Lebesgue measure in the set of square matrices \cite{gleizer2021computing}, and as such these matrices can be considered generic, or non-pathological. 

Proposition \ref{prop:quadsubspace} presents a simple way to verify whether a linear subspace is a subset of a quadratic cone.	
\begin{prop}[\!\!\cite{gleizer2021hscc}]\label{prop:quadsubspace}
Let $\As$ be a linear subspace with basis $\vv_1, \vv_2,..., \vv_m$, and let $\Vm$ be the matrix composed of the vectors $\vv_i$ as columns. Let $\Qm \in \S^n$ be a symmetric matrix and define $\Qs_n \coloneqq \{\xv \in \R^n \mid  \xv\tran\Qm\xv \geq 0\}$, $\Qs_s \coloneqq \{\xv \in \R^n \mid  \xv\tran\Qm\xv > 0\}$ and $\Qs_e \coloneqq \{\xv \in \R^n \mid  \xv\tran\Qm\xv = 0\}$. Then, $\As \setminus \{\O\} \subseteq \Qs_n$ (resp.~$\Qs_s$ and $\Qs_e$) if and only if $\Vm\tran\Qm\Vm \succeq \O$ (resp.~$\Vm\tran\Qm\Vm \succ \O$ and $\Vm\tran\Qm\Vm = \O$).
\end{prop}

\section{Event-triggered control and its traffic}\label{sec:prob}

Consider a closed-loop linear system
\begin{align}
\dot{\xiv}(t) &= \Am\xiv(t) + \Bm\Km\hat{\xiv}(t),\label{eq:plant}\\
\xiv(0) &= \hat{\xiv}(0) = \xv_0, \nonumber
\end{align}
which is a sampled-data state feedback with zero-order hold: the state $\xiv(t) \in \R^\nx$ is sampled at instants $t_i$, $\forall i \in \N$, and held constant for feedback, which makes the state signal used for control $\hat\xiv$ satisfy $\hat\xiv(t) = \xiv(t_i), \forall t \in [t_i, t_{i+1})$. Matrices $\Am, \Bm, \Km$ have appropriate dimensions.

In ETC, a \emph{triggering condition} determines the sequence of times $t_i$. In PETC, this condition is checked only periodically, with a fundamental checking period $h$. The sampling time $t_{i+1}$ hence assumes the following form:
\begin{equation}\label{eq:trigtime}
t_{i+1} = \inf\{t \in \Ts \mid t > t_i \text{ and } c(t-t_i, \xiv(t), \hat\xiv(t))\},
\end{equation}
where $c : \Ts \times \R^\nx \times \R^\nx \to \{\text{true},\text{false}\}$ is the \emph{triggering condition}, and $\Ts$, the set of checking times, is $\R_+$ for continuous ETC (CETC) and $h\N$ for PETC.
We consider the family of \emph{quadratic triggering conditions} from \cite{heemels2013periodic} with an additional maximum inter-event time condition below:
\begin{equation}\label{eq:quadtrig}
c(s, \xv, \hat\xv) \coloneqq 
\begin{bmatrix}\xv \\ \hat\xv \end{bmatrix}\tran
\!\Qm(s) \begin{bmatrix}\xv \\ \hat\xv\end{bmatrix} > 0
\text{ or } \ s \leq \bar{\tau}
\end{equation}
where $\Qm : \Ts \to \S^{2\nx}$ is the designed triggering matrix function (possibly constant), and $\bar{\tau}$ is the chosen maximum inter-event time.%
\footnote{Often a maximum inter-event time arises naturally from the closed-loop system itself (see \cite{gleizer2018selftriggered}). Still, one may want to set a smaller maximum inter-event time so as to establish a ``heart beat'' of the system.}
When $\Ts = \R_+$, we assume $\Qm$ is differentiable. %
Many of the triggering conditions available in the literature can be written as in Eq.~\eqref{eq:quadtrig}; the interested reader may refer to \cite{heemels2013periodic} for a comprehensive list of triggering and stability conditions.

We are interested in modeling the traffic generated by (P)ETC, i.e., understanding how the inter-sample times evolve from different initial conditions. As noted in \cite{gleizer2020scalable}, the inter-event time $t_{i+1} - t_{i}$ is solely a function of the $i$-th sample $\xiv(t_i)$. First, note that, $\xiv(t)$ is a function of $\hat\xiv(t) = \xiv(t_i)$ and the elapsed time $s \coloneqq t-t_i$:
\begin{gather}
\xiv(t_i+s) = \Mm(s)\xiv(t_i), \label{eq:intersamplestate}\\
\!\!\Mm(s) \coloneqq \Am_\d(s) + \Bm_\d(s)\Km \coloneqq \e^{\Am s} + \int_0^{s}\!\e^{\Am t}\d t \Bm\Km.\!\!\nonumber
\end{gather}
Now let $\tau : \R^\nx \to (0,\bar\tau] \cap \Ts$ be the inter-event time function.\footnote{We assume the triggering condition prevents Zeno behavior, which is standard in ETC design.} That is, for every state $\xv \in \R^\nx$, $\tau$ must return the value of $t_{i+1}-t_i$. It follows from Eqs.~\eqref{eq:trigtime}--\eqref{eq:intersamplestate} that
\begin{equation}\label{eq:time}
\begin{gathered}
	\tau(\xv) = \inf\left\{s \in \Ts \mid \xv\tran\Nm(s)\xv > 0 \text{ or } s=\bar{\tau}\right\}, \\
	\Nm(s) \coloneqq \begin{bmatrix}\Mm(s) \\ \I\end{bmatrix}\tran
	\Qm(s) \begin{bmatrix}\Mm(s) \\ \I\end{bmatrix},
\end{gathered}
\end{equation}
where $\I$ denotes the identity matrix. Thus, the event-driven evolution of sampled states can be compactly described by the recurrence
\begin{equation}\label{eq:samples}
\xiv(t_{i+1}) = \Mm(\tau(\xiv(t_i)))\xiv(t_i).
\end{equation}
Throughout the paper, we refer to the function above as the \emph{sample system}, using the shortened version
\begin{equation}\label{eq:samplemap}
\begin{aligned}
	\xv_{i+1} &= f(\xv_i), \\
	y_i &= \tau(\xv_i),
\end{aligned}
\end{equation}
where $\xv_i \coloneqq \xiv(t_i)$ and $f(\xv) \coloneqq \Mm(\tau(\xv))\xv$. The map is equipped with an output $y$ which is the associated inter-event time: for a traffic model, this is the output of interest. We shall denote the sequence of outputs from Eq.~\eqref{eq:samplemap} for a given initial state $\xv_0$ by $\{y_i(\xv_0)\}$. 

\subsection{Isochronous subsets in ETC}

We start our analysis of sampling behaviors of ETC by studying the subsets of $\R^\nx$ that generate the same inter-sample time. The first characteristic to be highlighted is that inter-sample times are insensitive to magnitude.

\begin{prop}[Adapted from \cite{mazo2010iss}]\label{prop:props}
The sample system \eqref{eq:samplemap} is homogeneous; more specifically, for all $\lambda \in \R \setminus \{0\}, \tau(\lambda\xv) = \tau(\xv)$ and $f(\lambda \xv) = \lambda f(\xv)$.
\end{prop}
\begin{proof}
With respect to Eq.~\eqref{eq:time}, $\sign((\lambda\xv)\tran\Qm(s)(\lambda\xv)) = \sign(\lambda^2\xv\tran\Qm(s)\xv) = \sign(\xv\tran\Qm(s)\xv),$ hence $\tau(\lambda\xv) = \tau(\xv).$ With this, $f(\lambda\xv) = \Mm(\tau(\lambda\xv))\lambda\xv = \lambda\Mm(\tau(\xv))\xv = \lambda f(\xv).$
\end{proof}

This fact implies that the sequence $\{y_i(\xv)\}$ is equal to $\{y_i(\lambda\xv)\}$, for any $\lambda \neq 0.$ Hence, to determine whether ETC exhibits fixed (periodic) behavior, we need to verify which lines passing thorough the origin, or collections of lines, 
are invariant under $f$ or under a finite iterate of $f$. Hereafter we shall refer to lines that pass through the origin as \emph{o-lines}.

Let us first look in detail what are the subsets of $\R^\nx$ which share the same inter-event time:
\begin{defn}
Consider system \eqref{eq:samplemap}. We denote by $\Qs_s \subset \R^\nx$, the set of all states which trigger after $s$ time units, i.e.,
$$ \Qs_s \coloneqq \{\xv \in \R^\nx \mid \tau(\xv) = s\}. $$
We call $\Qs_s$ an \emph{isochronous subset}.\footnote{The concept of isochronous manifolds was introduced in ETC for nonlinear homogeneous systems in \cite{anta2011exploiting}. In our case, the isochronous subsets of CETC are (in general) $(\nx-1)$-dimensional subsets, but may not be manifolds. Whether they are manifolds or not is not relevant to our results.} 
\end{defn}
\begin{prop}\label{prop:isochronous}
Consider system \eqref{eq:plant}--\eqref{eq:quadtrig}. An isochronous subset $\Qs_s$ can be characterized as
\begin{enumerate}[i)]
\item If $\Ts = \R_+$ (CETC) and $s < \bar{\tau}$, $\Qs_s = \{\xv \in \R^n \mid \xv\tran\Nm(s)\xv = 0 \text{ and } \xv\tran\Nm(s')\xv \leq 0, \forall s' < s \text{ and }  \xv\tran\dot\Nm(s)\xv > 0\}$.
\item If $\Ts = h\N$ (PETC) and $s < \bar{\tau}$, $\Qs_s = \{\xv \in \R^n \mid \xv\tran\Nm(s)\xv > 0 \text{ and } \xv\tran\Nm(s')\xv \leq 0, \forall s' < s, s' \in h\N\}$.
\item $\Qs_{\bar\tau} = \{\xv \in \R^n \mid \xv\tran\Nm(s')\xv \leq 0, \forall s' < \bar\tau, s' \in \Ts\}$.
\end{enumerate}
\end{prop}
\begin{proof}
This is a trivial manipulation of Eq.~\eqref{eq:time}, where in (i) we use the fact that $\Nm(s)$ is differentiable over $[0,\bar\tau)$.
\end{proof}

The isochronous subset $\Qs_{\bar\tau}$ is the intersection of an algebraic set with infinitely many semialgebraic sets for CETC; for PETC, it is the intersection of finitely many semialgebraic sets. We can extend the definition of isochronous subset to a sequence of inter-sample times:
\begin{defn}[Isosequential subset]\label{def:isosequential}
Consider the system \eqref{eq:samplemap}. The set $\Qs_{y\sigma}$, $y \in \Ts, \sigma \in \Ts^{m-1}$ for some $m\in\N$, is defined recursively as the set of states $\xv \in \R^\nx$ such that $\xv \in \Qs_{y_1},$ and $\Mm(y_1)\xv \in \Qs_{\sigma}$. By convention, $\Qs_{\epsilon} = \R^\nx,$ where $\epsilon$ denotes the empty sequence.
\end{defn}
As we can see, the set $\Qs_\sigma$ is also the intersection of (semi)algebraic sets as in the singleton case. 
We end this section with a result that simplifies the analysis for CETC under some special conditions.

\begin{prop}\label{prop:cont}
Consider system \eqref{eq:samplemap} and $\Ts = \R_+$ (CETC). If $\bar\tau = \inf\{s > 0 \mid \Nm(s) \succ \O\} < \infty$ and
$$ \forall \xv \in \R^\nx \setminus \{\O\}, s\in (0,\bar\tau], \xv\tran\Nm(s)\xv = 0 \implies \xv\tran\dot\Nm(s)\xv > 0 $$
then 
\begin{enumerate}[i)]
\item $f$ and $\tau$ are differentiable;
\item $\Qs_s = \{\xv \in \R^\nx \mid \xv\tran\Nm(s)\xv = 0\}, \forall s \in (0, \bar\tau).$
\end{enumerate}
\end{prop}
\begin{proof}
We first prove (ii), which is a lemma to (i).

ii) Consider the function $\phi_{\xv}(s) = \xv\tran\Nm(s)\xv$, which is differentiable. We want to prove that, if for some $s$, $\phi_{\xv}(s)=0$, then all conditions from Prop.~\ref{prop:isochronous} (i) are satisfied; since $\dot\phi_{\xv}(s) > 0$ by assumption, we need to prove that $\phi_{\xv}(s') \leq 0, \forall s' < s$. Now, since $\phi_{\xv}(s) = 0 \implies \dot\phi_{\xv}(s) > 0,$ from continuity, it holds that $\phi_{\xv}(s^-) < 0$ for some $s^-<s$. For contradiction, assume $\phi_{\xv}(s') > 0$ for some $s' < s^-$. Then, from Bolzano's theorem there is some point $s'' \in (s', s^-)$ such that that $\phi_{\xv}(s'') = 0.$ One such $s''$ must have $\phi_{\xv}(s'')$ cross zero from positive to negative, which implies $\dot\phi_{\xv}(s'') \leq 0$, leading to a contradiction.

i) Now $\tau(\xv)$ is characterized by the implicit equation $\xv\tran\Nm(\tau)\xv = 0$. Therefore we can simply apply the implicit function theorem, whose condition ($\phi_{\xv}(s)=0\Rightarrow\dot\phi_{\xv}(s)\neq0$) is satisfied by ours.
\end{proof}

\begin{rem}
The condition in Proposition \ref{prop:cont} is equivalent, by the s-procedure, to the linear matrix inequality $\exists \lambda \in \R: \lambda\Nm(s) + \dot\Nm(s) \succ \O$. Note that it is trivially satisfied if $\dot\Nm(s) \succ \O$ for all $s\in [0,\bar\tau]$, which holds when the triggering function $\phi_{\xv}$ is monotonically increasing for all $\xv$.
\end{rem}

The condition in Proposition \ref{prop:cont} ensures that the triggering function crosses zero only once for each initial condition $\xv$, which in turn simplifies the isochronous subset description to a simple quadratic form and renders $f$ and $\tau$ continuous. As we will see, even when this continuity is observed, the behaviors generated by ETC can be extremely rich.

\subsection{Problem statement}\label{ssec:prob}

We are interested in quantifying the traffic usage of system \eqref{eq:plant}--\eqref{eq:quadtrig}, which involves studying the sample system \eqref{eq:samplemap}. Some candidate metrics are the following:
\begin{itemize}
\item $\text{Inf} \coloneqq \inf_{\xv\in\R^\nx} \tau(\xv)$;
\item $\text{Sup} \coloneqq \sup_{\xv\in\R^\nx} \tau(\xv)$;
\item $\text{InfLimInf} \coloneqq \inf_{\xv\in\R^\nx}\liminf_{i\to\infty} y_i(\xv)$;
\item $\text{SupLimSup} \coloneqq \sup_{\xv\in\R^\nx}\limsup_{i\to\infty} y_i(\xv)$;
\item $\text{InfLimAvg}\coloneqq\inf_{\xv\in\R^\nx}\liminf_{n\to\infty}\frac{1}{n+1}\sum_{i=0}^{n}y_i(\xv).$
\item $\text{SupLimAvg}\coloneqq\sup_{\xv\in\R^\nx}\limsup_{n\to\infty}\frac{1}{n+1}\sum_{i=0}^{n}y_i(\xv).$
\end{itemize}

The first two metrics are simply the minimal and maximal inter-event times that can be exhibited. The minimal is the one that has received most attention in the literature, mainly to prove absence of Zeno behavior for different triggering conditions. These metrics serve as worst- and best-case inter-event times and provide a basic information about how sample-efficient one given ETC system is. Inf is trivially calculated as $\text{Inf} = \inf\{s \in \Ts \mid \Nm(s) \nprec \O\}$, while Sup is a bit more complicated: $\text{Sup} = \min(\bar\tau, \inf\{s \in \Ts \mid \forall \xv \in \R^\nx \exists s'<s \ \xv\tran\Nm(s')\xv > \O\})$.
The last four metrics concern limit behaviors of the system. InfLimInf gives what is the minimal inter-sample time the system can exhibit as the number of samples goes to infinity: in other words, after transients on the sequence $y_i$ vanish. The symmetric case value is given by SupLimSup. Finally, InfLimAvg (SupLimAvg) gives the minimum (maximum) among initial states of average inter-sample time. Here, $\liminf$ ($\limsup$) is used to ensure that the value exists even if the sequence of averages does not converge.

We argue that the limit metrics are more informative to determine the performance of a sampling mechanism than the simpler Inf and Sup metrics. For instance, if the states $\xv$ associated to Inf are transient, in the sense that from almost all other initial states they are never visited, the Inf metric turns out to be very conservative; after a few samples, the typical inter-sample time of the system will be higher. InfLimInf gives the complementary information of what minimal inter-sample time can appear infinitely often. InfLimAvg informs about the average utilization rate. A disadvantage of these two metrics is that they can still capture exceptional behavior: suppose, for example, that a measure-zero set $\Xs \subset \R^\nx$ is invariant under \eqref{eq:samplemap} and it is associated to the InfLimInf or InfLimAvg of the system; moreover, suppose for every state $\xv \notin \Xs$, the trajectories $\xiv_{\xv}(t_i), i \in \N,$ never enter $\Xs$, but instead converge to some other subset with higher values of InfLimInf or InfLimAvg. Then, the metric will not reflect the dominant performance of the system. This information might still be useful, but a more robust version of these metrics is of interest. In any case, robust or not, we need a hint of how one could compute these metrics. This will allow us to properly define what robust should be in this context.

{\bf Problem Statement.} Given an ETC system, (i) identify its limit traffic patterns, (ii) characterize their robustness w.r.t.~small perturbations in the initial state, and (iii) compute the system's robust limit metrics.

\section{Qualitative analysis: limit behaviors in ETC}\label{sec:behaviors}

In this section we investigate the limit behaviors of the traffic generated by ETC. We first see that limit metrics are insensitive to transient behavior; then we look at some examples to classify the different limit behaviors that can be exhibited. 
In several cases, ETC traffic converges to a periodic sampling pattern, which is shown to be characterized by linear invariants. This characterization allows us to show that, if PETC stabilizes a periodic traffic pattern, then this traffic pattern can be used as a sampling schedule that guarantees GES of the system.

\subsection{Properties of limit metrics}

The following trivial result shows that limit metrics are insensitive to transient behavior. We focus on inferior metrics, as the superior counterparts follow similar reasoning.

\begin{prop}\label{prop:limmetrics}
Let $\{k_i\}$ be a sequence of real numbers and decompose it as $k_i = a_i + b_i$, where $b_i$ is the \emph{transient component}, i.e., it satisfies $\lim_{i\to\infty} b_i = 0$. Then,
\begin{enumerate}[(i)]
\item $\liminf_{i\to\infty}k_i = \liminf_{i\to\infty}a_i,$
\item $\liminf_{n\to\infty}\frac{1}{n+1}\sum_{i=0}^{n}k_i = \liminf_{n\to\infty}\frac{1}{n+1}\sum_{i=0}^{n}a_i.$
\end{enumerate}
\end{prop}

\begin{proof}
It is a property of $\liminf$ that $\liminf_{i\to\infty}(a_i + b_i) = \liminf_{i\to\infty}a_i + \liminf_{i\to\infty}b_i$ if either $\{a_i\}$ or $\{b_i\}$ converge. Thus, result (i) trivially holds. For item (ii), we only need to prove that the sequence $\{\frac{1}{n+1}\sum_{i=0}^{n}b_i\}$ converges and is equal to zero. For this, we apply the Stolz--Cesàro theorem:
$$ \liminf_{n\to\infty}b_i = 0 \leq \liminf_{n\to\infty}\frac{1}{n+1}\sum_{i=0}^{n}b_i \leq \limsup_{n\to\infty}b_i = 0, $$
which concludes the proof.
\end{proof}

\begin{cor}\label{cor:limavgperiodic}
Let $\{k_i\}$ be \emph{ultimately periodic}, i.e., $k_i = a_i + b_i$, $\lim_{i\to\infty} b_i = 0$ and $a_{i+M} = a_i$ for some $M \in \N_+$ and all $i$. Then,
\begin{enumerate}[(i)]
\item $\liminf_{i\to\infty}k_i = \min_{i<M}a_i,$
\item $\liminf_{n\to\infty}\frac{1}{n+1}\sum_{i=0}^{n}k_i = \frac{1}{M}\sum_{i=0}^{M-1}a_i.$
\end{enumerate}
\end{cor}

Proposition \ref{prop:limmetrics} implies that computing limit metrics of ETC is fundamentally a problem of finding its limit behaviors, ignoring transients. In particular, given Corollary \ref{cor:limavgperiodic}, if a sequence of inter-event times $y_i$ converges to a periodic pattern, then the limit metrics are solely functions of the periodic component. This motivates us to study fixed and periodic solutions of \eqref{eq:samplemap}; for example, if some $y$ is a recurring pattern of \eqref{eq:samplemap}, then there must be a subset of $\Qs_y$ that is invariant. This is done in §\ref{ssec:invariants}. Before that, we investigate some examples to understand what are the possible limit behaviors exhibited by ETC.

\subsection{An illustrative example}\label{ssec:example}

Consider system \eqref{eq:plant}--\eqref{eq:quadtrig} with $\nx=2$. In this case, an o-line is uniquely defined by the angle $\theta \coloneqq \ang(\xv) \coloneqq \arctan{x_1/x_2} \in [-\pi/2, \pi/2)$. Using the coordinate $\theta$ and identifying points along an o-line (that is, regarding any point along an o-line as the same), the sample system \eqref{eq:samplemap} becomes
\begin{equation}\label{eq:thetamap}
\begin{aligned}
\theta_{i+1} &= \tilde{f}(\theta_i) \coloneqq \ang{\left(f(\sincos)\right)}, \\
y_i &= \tilde\tau(\theta) \coloneqq \tau\big(\sincos\big).
\end{aligned}
\end{equation}
The map $\tilde{f}$ can be seen as a map on the unit circle. An analysis of system \eqref{eq:thetamap} has been conducted in \cite{rajan2020analysis}, aiming at finding fixed points or the absence of them. In the cases studied in \cite{rajan2020analysis}, when there was a fixed point, there was always a stable fixed point. In the next example we show that this is not always true, and investigate the many possible behaviors that ETC traffic exhibits.

\begin{example}\label{ex:r2}
Consider system \eqref{eq:plant}--\eqref{eq:quadtrig} with
\begin{equation}\label{eq:example2d}
\begin{gathered}
\Am = \begin{bmatrix}0 & 1 \\ -2 & 3\end{bmatrix}, \ \Bm = \begin{bmatrix}0 \\ 1\end{bmatrix}, \\
c(s,\xv,\hat\xv) = |\xv - \hat\xv| > \sigma|\xv|,
\end{gathered}
\end{equation}
where $\sigma \in (0,1)$ is the triggering parameter. This is the seminal triggering condition of \cite{tabuada2007event}, which can be put in the form \eqref{eq:quadtrig} with sufficiently large $\bar\tau$. The graphs of $\tilde{f}$ and $\tau(\cdot)$, for CETC ($\Ts = \R_+$) are given for four cases:

\begin{enumerate}
\item $\Km = \begin{bmatrix}0 & -5\end{bmatrix}, \sigma=0.2$: Fig.~\ref{fig:ex1c1}. This map is invertible, orientation-preserving\footnote{A map $f: \Xs \to \Xs$ is said to be \emph{orientation-preserving} if its Jacobian $J_f$ satisfies $\det(J_f(\xv)) > 0$ for all $\xv \in \Xs$.}, and has no fixed points.  
\item $\Km = \begin{bmatrix}0 & -6\end{bmatrix}, \sigma=0.32$: Fig.~\ref{fig:ex1c2}. This map is no longer invertible. It has one unstable fixed point near $\theta=-1.3$ and one stable fixed point near $\theta=-0.6$.
\item $\Km = \begin{bmatrix}0 & -6\end{bmatrix}, \sigma=0.5$: Fig.~\ref{fig:ex1c3}. This map has two unstable fixed points, but a stable period-4 solution as indicated by the cobweb diagram.
\item $\Km = \begin{bmatrix}0 & -6\end{bmatrix}, \sigma=0.6$: Fig.~\ref{fig:ex1c4}. This map has no stable fixed points or orbits, and exhibits chaotic behavior. By inspection of the graph, the system has as a minimal set\footnote{A minimal set is an invariant set which contains no proper subsets that are also invariant.} the interval $[-1.07, -0.42]$, which contains the maximum inter-sample time $\bar\tau \approxeq 0.76$, so SupLimSup = Sup $\approxeq$ 0.76.
\end{enumerate}

\begin{figure}
\begin{center}
\input{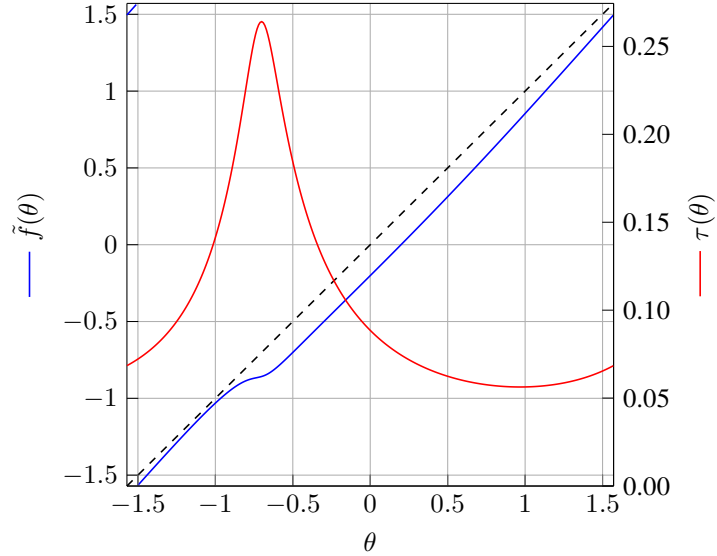}
\caption{\label{fig:ex1c1} Map $\tilde{f}$ and inter-event time $\tau$ for case 1 of Example \ref{ex:r2}.}
\end{center}
\end{figure}

\begin{figure*}[t!]
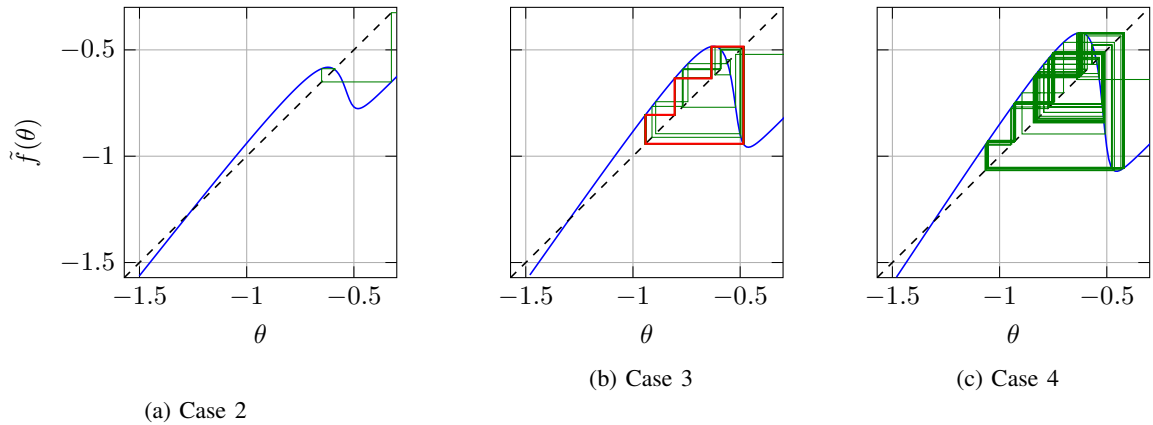

\begin{subfigure}[t]{0.32\linewidth}
\centering
\input{figures/ex1_case2.tex}\hspace{-2em}
\caption{Case 2}
\label{fig:ex1c2}
\end{subfigure}~~~~~
\begin{subfigure}[t]{0.32\linewidth}
\centering
\input{figures/ex1_case3.tex}
\caption{Case 3}
\label{fig:ex1c3}
\end{subfigure}\hspace{-1.2em}
\begin{subfigure}[t]{0.32\linewidth}
\centering
\input{figures/ex1_case4.tex}
\caption{Case 4}
\label{fig:ex1c4}
\end{subfigure}
\caption{Maps $\tilde{f}$ for Example \ref{ex:r2}, along with cobweb diagrams of solutions of \eqref{eq:thetamap} starting from $\theta_0 = 0.$  A stable orbit for Case 3 is highlighted in red.}
\end{figure*}

Finally, notice that all these maps are differentiable, but this is not always the case, as has been observed in \cite{rajan2020analysis}. In particular, it is almost never the case for PETC ($\Ts=h\N$). One example is shown in Fig.~\ref{fig:ex1c5}, for $\Km = \begin{bmatrix}0 & -6\end{bmatrix}, \sigma=0.32$ (like case 2) and $h=0.05$. Different from the CETC case, its fixed points are unstable and it exhibits chaos.

\begin{figure}
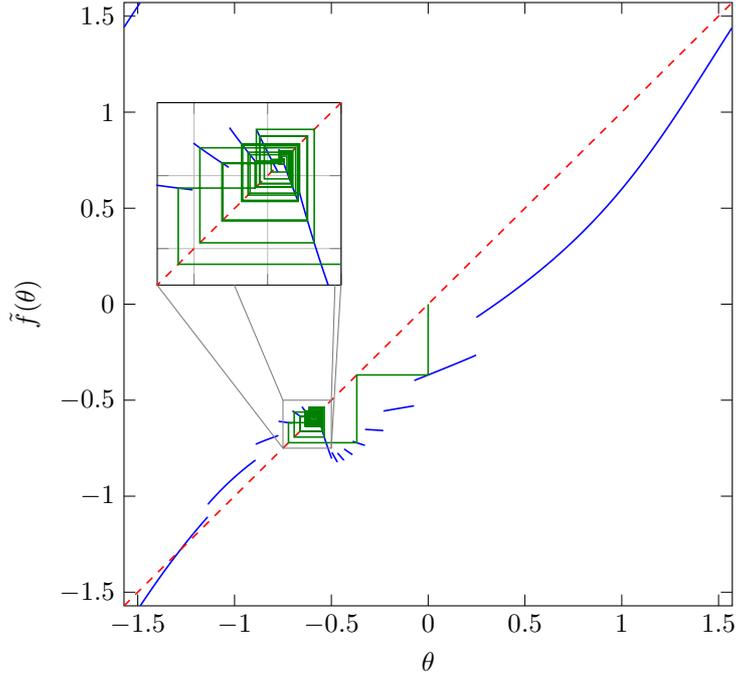

\begin{center}
\begin{tikzpicture}[spy using outlines={rectangle, magnification=5, connect spies}]

\begin{axis}[
scale only axis,
height=0.5\linewidth,
width=0.5\linewidth,
x grid style={white!69.0196078431373!black},
xlabel={$\theta$},
xmin=-1.571, xmax=1.571,
xtick style={color=black},
y grid style={white!69.0196078431373!black},
ylabel={$\tilde{f}(\theta)$},
ymin=-1.571, ymax=1.571,
ytick style={color=black}
]

\input{figures/ex1_case2_PETC_data.tex}

\draw[draw=gray] (-0.75, -0.75) rectangle (-0.5, -0.5);

\coordinate (ul) at (-0.75, -0.5);
\coordinate (ur) at (-0.5, -0.5);
\coordinate (bl) at (-0.75, -0.75);
\coordinate (br) at (-0.5, -0.75);

\coordinate (pt) at (axis cs:-1.4,0.1);

\end{axis}

\begin{axis}[
	at=(pt),
	xmin=-0.75,xmax=-0.5,
	ymin=-0.75,ymax=-0.5,
	width = 4cm,
	height=4cm,
	xmajorgrids, ymajorgrids,
	axis background/.style={fill=white},
	tick label style={font=\footnotesize} ,
	yticklabels={},
	xticklabels={},
	]
	
	\input{figures/ex1_case2_PETC_data.tex}
	
	\coordinate (ul2) at (-0.75, -0.5);
	\coordinate (ur2) at (-0.5, -0.5);
	\coordinate (bl2) at (-0.75, -0.75);
	\coordinate (br2) at (axis cs:-0.5, -0.75);
\end{axis}

	\coordinate (r) at (intersection of bl2--br2 and ur--ur2);
	\coordinate (l) at (intersection of bl2--br2 and ul--ul2);
	\draw[draw=gray] (br) -- (br2);
	\draw[draw=gray] (bl) -- (bl2);
	\draw[draw=gray] (ur) -- (r);
	\draw[draw=gray] (ul) -- (l);

\end{tikzpicture}
\caption{\label{fig:ex1c5} Map $\tilde{f}$ for the PETC implementation of case 2 of Example \ref{ex:r2}, with $h=0.05$, along with a cobweb diagram of a solution of \eqref{eq:thetamap} starting from $\theta_0 = 0$.}
\end{center}
\end{figure}
\end{example}

\begin{rem}
Invertible orientation-preserving maps on the circle have been extensively studied in the field of dynamical systems \cite{demelo2012one}, and they have an attribute called \emph{rotation number}. 
When the rotation number is rational $p/q$, $p$ and $q$ coprime, all solutions converge to a periodic orbit of period $q$. When it is irrational, all solutions are quasi-periodic: oscillatory, but the same point is never visited twice.
In the latter case, if $\tilde{f}$ is twice continuously differentiable, it is topologically conjugate to an irrational rotation $g(\theta) = (\theta + r\uppi \mod \uppi) - \uppi/2,$ which is ergodic and its orbit is dense in $[-\uppi/2, \uppi/2)$. Hence, InfLimInf = Inf, and InfLimAvg = SupLimAvg can be obtained to arbitrary precision through simulations from any initial condition.
\end{rem}

\subsection{Invariant isosequential sets in ETC}\label{ssec:invariants}

Example \ref{ex:r2} illustrates the complex behavior that can emerge in ETC traffic. Nonetheless it becomes apparent that obtaining fixed or periodic patterns is a fundamental step in the traffic characterization. The first thing we want is a computational or analytical method to determine fixed and periodic patterns. Then, we want to characterize their local stability.

In \cite{gleizer2021computing}, it has been shown that periodic patterns can be characterized by linear invariants. 
\begin{thm}\label{thm:verifycycle}(\!\!\cite{gleizer2021computing})	Consider system \eqref{eq:samplemap}, let $\sigma \coloneqq y_1y_2...y_m$ be a sequence of $m$ outputs. Denote by $\Mm_\sigma \coloneqq \Mm(y_m)\cdots\Mm(y_2)\Mm(y_1)$. (i) If $\Mm_\sigma$ is nonsingular and there exists a linear invariant $\As$ of $\Mm_\sigma$ such that $\As \setminus \{\O\} \subseteq \Qs_\sigma$, then $\sigma^\omega$ is a possible output sequence of system \eqref{eq:samplemap}. 
Moreover, if (ii) $\Mm_\sigma$ is additionally mixed and of irrational rotations, then $\sigma^\omega$ being an output sequence of system \eqref{eq:samplemap} implies that there exists a linear invariant $\As$ of $\Mm_\sigma$ such that $\As \subseteq \cl(\Qs_\sigma).$ 
\end{thm}
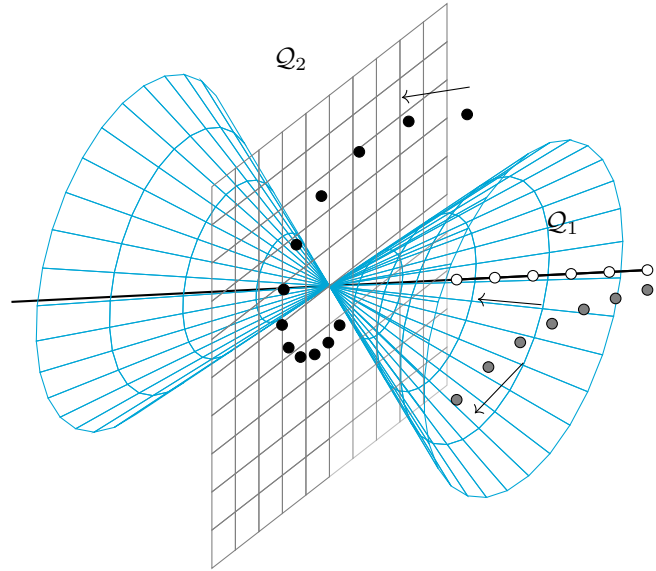
\begin{figure}
\centering
\begin{tikzpicture}
	\begin{axis}[
		hide axis,
		width=0.8\linewidth,
		domain=-1:1,
		y domain=0:-2*pi,
		xmin=-1.5, xmax=1.5,
		ymin=-1.5, ymax=1.5, zmin=-1.2,
		samples=10,
		samples y=30,
		]
		\addplot3[domain=-1.5:1.5, samples y=1, thick]
		({x}, {0.2*x}, {0.2*x});
		
		\addplot3[mesh, domain=0:1, tudcyan, samples=5]
		({x},{1.1*x*cos(deg(y))},{1.1*x*sin(deg(y))});
		\addplot3[mesh, tudcyan, domain=-1:0, y domain = pi/5:3*pi/2, samples y=20, samples=5,]
		({x},{1.1*x*cos(deg(y + pi))},{1.1*x*sin(deg(y + pi))});
		
		\node[draw=none] at (1.2,0,0.7) (q1) {$\Qs_1$};
		\node[draw=none] at (-0.2,0,1.5) (q2) {$\Qs_2$};
		
		\addplot3[domain=0.6:1.5, samples y=1, thick]
		({x}, {0.2*x}, {0.2*x});
		
		\addplot3[only marks, domain=0.6:1.5, samples y=1, samples=6, mark=*,mark options={fill=white}]
		({x}, {0.2*x}, {0.2*x});
		\addplot3[only marks, domain=0.6:1.5, samples y=1, samples=7, mark=*,mark options={fill=gray}]
		({x}, {0.2*x}, {0.2*x - exp(-2*(x-0.5))});
		
		\addplot3[mesh, domain=-1.3:1.3, y domain=0:1.3, samples=11, samples y=6, thin, gray]
		({0}, {x}, {y});
		\addplot3[mesh, domain=-1.3:0, y domain=-1.3:0, samples=6, samples y=6, thin, gray]
		({0}, {x}, {y});
		\addplot3[mesh, domain=0:1.3, y domain=-1.3:0, samples=6, samples y=6, thin, gray, opacity=0.5]
		({0}, {x}, {y});
		
		\addplot3[only marks, domain=pi/4:1.6*pi, samples y=1, samples=12, mark=*]
		({0.7*exp(-x)}, {1.5*cos(deg(x))*exp(-0.3*(x))}, {1.5*sin(deg(x))*exp(-0.3*(x))});
		
		\draw[->] (1,0.2,0) -- (0.7, 0.14, 0);
		\draw[->] (1,0,-0.3) -- (0.75, 0, -0.7);
		\draw[->] (0.35,0.8,1.05) -- (0, 0.8, 0.9);
		
	\end{axis}
\end{tikzpicture}
\caption{\label{fig:modesandcones} Illustration of Theorem \ref{thm:verifycycle} in $\R^3$. The blue cone splits $\R^3$ into $\Qs_1$ and $\Qs_2$ the line is an invariant of $\Mm(1)$ and the plane is an invariant of $\Mm(2).$ Points indicate distinct sample trajectories $\{\xv_i\}$, with the arrows indicating progress of time.}
\end{figure}

According to Theorem \ref{thm:verifycycle}, ETC exhibits a periodic sampling pattern whenever a linear invariant of the corresponding linear system is contained in the associated isosequential subset; in fact, the set $\As \setminus \{\O\}$ is a periodic set (with period $m$) of $f$. %
An illustration for a PETC system with $\nx=3$ and $\bar{k}=2$ is given in Fig.~\ref{fig:modesandcones}: because an invariant of $\Mm(1)$ is a subset of $\Qs_1$, we know that $1^\omega$ is a sampling pattern exhibited by the system; likewise with $\Mm(2)$. %
The corollary given below (see the proof in the Appendix) states that in general this invariant is an o-line (or \emph{o-plane}, a plane through the origin), and we have an if-and-only-if condition.
\begin{cor}\label{cor:lineandplane}
Given the premises of Theorem \ref{thm:verifycycle}, assume (i) $\Mm_\sigma$ is nonsingular, mixed, and of irrational rotations, and that (ii) for every linear invariant $\As$ of $\Mm_\sigma$, $\As \subseteq \cl(\Qs_\sigma) \implies \As \setminus \{\O\} \subseteq \Qs_\sigma$. Then $\sigma^\omega$ is a possible output sequence of system \eqref{eq:samplemap} if and only if there exists an o-line or o-plane $\As$ invariant of $\Mm_\sigma$ such that $\As \setminus \{\O\} \subseteq \Qs_\sigma.$
\end{cor}

Condition (ii) is satisfied in the illustrative example of Fig.~\ref{fig:modesandcones}, as the invariants lie in the interior of the corresponding isochronous sets. This result is particularly useful to verify whether a given periodic sequence is exhibited by the system \eqref{eq:samplemap}, and is instrumental in the symbolic methods used in §\ref{sec:symbolic}.

\begin{rem}\label{rem:cetcfixed} Using Corollary \ref{cor:lineandplane}, one can find fixed inter-sample patterns $(t)^\omega$ by searching over $t \in [\underline\tau, \bar\tau]$ for an $\Mm(t)$ with a linear subspace belonging to $\Qs_t$, which can be checked using Prop.~\ref{prop:quadsubspace}. 
This search is one-dimensional, in contrast to the search for invariants of system \eqref{eq:samplemap} over $\R^\nx$.
\end{rem} 

The following lemma is useful when dealing with fixed o-lines.
\begin{lem}\label{lem:fixedoline}
Let $\lv$ be a fixed o-line of $f$ in system \eqref{eq:samplemap}, i.e., $\xv in \lv \implies f(\xv) \in \lv$. Then, there exists a real $\lambda$ such that $f(\xv) = \lambda \xv$ for all $\xv \in \lv$.
\end{lem}
\begin{proof}
By Prop.~\ref{prop:props}, every $\xv \in \lv$ shares the same inter-sample time $\tau$. Then, $f(\xv) = \Mm(\tau)\xv = a(\xv)\xv$ since $f(\xv) \in \lv$. Hence, by definition of eigenvalues, $\xv$ is an eigenvector of $\Mm(\tau)$ and $a(\xv) = \lambda$ is the corresponding eigenvalue.
\end{proof}

For some classic triggering conditions, we can get some interesting specialized results:
\begin{prop}\label{prop:fixedraytabuada}
Consider system \eqref{eq:plant}--\eqref{eq:quadtrig} with $c(s,\xv,\hat\xv) \equiv |\xv - \hat\xv| > \sigma|\xv|,$ $\Ts=\R_+$, and assume $0 < \sigma < 1$ is designed rendering the closed-loop system GES. A fixed o-line with inter-sample time $\tau$ exists iff $1/(1+\sigma) \in \lambda(\Mm(\tau)).$ 
\end{prop}
\begin{proof}
By Lemma \ref{lem:fixedoline}, the points in the fixed o-line satisfy
$\hat\xv(t_{i+1}) = a\hat\xv(t_i)$, where $a$ is a real eigenvalue of $\Mm(\tau)$. 
From the triggering condition, it then holds that $|a\hat\xv - \hat\xv| = \sigma|a\hat\xv|$. Hence, $|a-1| = \sigma|a| \therefore a = 1/(1 \pm \sigma)$. Since $|a|<1$ for GES, $a = 1/(1 + \sigma) < 1.$ Because $\Mm(\tau)\hat\xv = a\hat\xv,$ $a\in\lambda(\Mm(\tau))$.
\end{proof}

\begin{prop}\label{prop:fixedraymazo}
Consider system \eqref{eq:plant}--\eqref{eq:quadtrig} with $c(s,\xv,\hat\xv) \equiv \xv\tran\Pm\xv > \e^{-2\rho s}{\hat\xv}\tran\Pm\hat\xv,$ $\Ts=\R_+$, with $0 < \rho < 1 $ and $\Pm \succ \O$.\footnote{This is the triggering condition initially used for STC in \cite{mazo2010iss}.} A fixed o-line with inter-sample time $\tau$ exists iff either $\e^{-\rho \tau}$ or $-\e^{-\rho \tau}$ is an eigenvalue of $\Mm(\tau).$ 
\end{prop}
\begin{proof}
Using the same arguments as in Prop.~\ref{prop:fixedraytabuada}, we have that $\hat\xv(t_{i+1}) = a\hat\xv(t_i)$. Let $\zv \coloneqq \sqrt{\Pm}\hat\xv$. An invariant o-line then satisfies $a^2\zv\tran\zv = \e^{-2\rho \tau}{\zv}\tran\zv \therefore a = \pm \e^{-\rho \tau}$. Now, $\sqrt{\Pm}^{-1}\Mm(\tau)\sqrt{\Pm}\hat\xv = \sqrt{\Pm}^{-1}\Mm(\tau)\zv = a\sqrt{\Pm}^{-1}\zv  = a\hat\xv.$ Since $\Mm(\tau)$ is similar to $\sqrt{\Pm}^{-1}\Mm(\tau)\sqrt{\Pm}$, $a\in\lambda(\Mm(\tau)).$
\end{proof}

A more general result can be obtained by invoking a result from topology (see the proof in the Appendix) to conclude about which cases a fixed o-line certainly exists, only by knowing the state-space dimension $\nx$; 
\begin{thm}\label{thm:topofixed}
Consider the system \eqref{eq:samplemap} and assume $f$ is continuous and $f(\xv) \neq \O$ for all $\xv\neq\O$. If $\nx$ is odd, then $f$ has a fixed o-line.
\end{thm}

Apart from o-lines, it is also interesting to know when can o-planes be fixed. A PETC example where this happens is illustrated in Fig.~\ref{fig:modesandcones}. %
The next result presents for which dimensions this can generally hold (the proof is also in the Appendix).

\begin{thm}\label{thm:dimplane}
System \eqref{eq:samplemap} can only exhibit a fixed o-plane $\Ps$ that is isochronous (i.e., $\forall \xv \in \Ps \setminus \{\O\}, \tau(\xv) = y$ for some $y$) if $\Nm(y)$ is singular or one of the following hold.
\begin{enumerate}[(i)]
\item $\nx=2$ and $\underline{\tau} = \bar{\tau}$ (periodic sampling, trivial);
\item $\nx = 3$ and $\Ts = h\N$ (PETC);
\item $\nx \geq 4$.
\end{enumerate}
\end{thm}

After having determined the fixed (or periodic) o-lines and o-planes of system \eqref{eq:samplemap}, the next step is to characterize their (local) attractivity. We say that an o-line $\lv \subset \R^\nx$ is attractive if for any other o-line $\lv'$ close enough to $\lv$, $\lim_{n\to\infty}f^n(\lv') = \lv$. The following can be applied for fixed o-lines (see proof in the Appendix.)

\begin{prop}\label{prop:attractivity}
Let $\lv \coloneqq \{a\xv \mid a \in \R \setminus \{0\}\}$ be a fixed o-line of system \eqref{eq:samplemap}, and suppose $f$ is differentiable at $\xv,$ with $J_f(\xv)$ being the corresponding Jacobian matrix. Take $\lambda$ as the real s.t.~$f(\xv) = \lambda\xv$ (Lemma \ref{lem:fixedoline}), and let $\Om_{\xv}$ be an orthonormal basis for the orthogonal complement of $\xv$. Then, if $\frac{1}{\lambda}\Om_{\xv}\, \tran J_f (\xv) \Om_{\xv}$ is Schur, then $\lv$ is locally attractive.
\end{prop}	

The Jacobian matrix can be expressed as $J_f = \partial (\Mm(\tau(\xv))\xv) / \partial \xv = \partial (\Mm(\tau(\xv)) / \partial \xv)\xv + \Mm(\tau(\xv))) = $
\begin{equation}\label{eq:jacobian}
\frac{-2}{\xv\tran\dot\Nm(\tau(\xv))\xv}\dot\Mm(\tau(\xv))\xv\xv\tran\Nm(\tau(\xv)) + \Mm(\tau(\xv)).
\end{equation}
The matrix $\Om_{\xv}\, \tran J_f (\xv) \Om_{\xv}$ is the Jacobian of $f$ w.r.t.~the non-radial directions and projected onto those. It is easy to see that the eigenvalues of $ \Om_{\xv}\, \tran J_f (\xv) \Om_{\xv}$ are the same as those of $J_f$ except the one associated with the eigenvector $\xv$, while $\lambda$ is precisely the eigenvalue associated with $\xv$; hence Prop.~\ref{prop:attractivity} gives a condition on the ratio between the largest-in-magnitude eigenvalue of $J_f$ and that of the fixed o-line in consideration. For fixed planes, this analysis may require more sophisticated analyses of orbital stability, such as Poincaré return maps.

As we see next, the case of PETC is revealing thanks to the fact that $\Mm$ is constant by parts and, thus, $J_f = \Mm(\tau(\xv))$ almost everywhere. 
Because PETC exhibits a discrete set of outputs, a proper definition of stability of an infinite sequence is necessary.

\begin{defn}\label{def:stablepetc}
Consider system \eqref{eq:samplemap} with $\Ts = h\N$ (PETC). An infinite sequence of outputs $\{y_i\}$ is said to be stable if there exists $\xv \in \R^\nx$ with a neighborhood $\Us$ such that every $\xv' \in \Us$ satisfies $y_i(\xv') = y_i(\xv) = y_i, \forall i \in \N.$
\end{defn}

\begin{prop}\label{prop:petchurwitz}
Consider system \eqref{eq:samplemap} with $\Ts=h\N$ (PETC) and assume it is GES. Let $\{y_i\}$ be a $p$-periodic output trajectory associated with it, and let $\Mm \coloneqq \Mm(y_{p-1})\cdots\Mm(y_1)\Mm(y_0)$. If 
$\{y_i\}$ is stable, then $\Mm$ is Schur.
\end{prop}
\begin{proof}
Every trajectory $\{\xv_i\}$ of \eqref{eq:samplemap} that generates $\{y_i\}$ satisfies $\xv_{i+p} = \Mm\xv_i$. If $\Mm$ is not Schur, then from almost every $\xv_0,$ (and hence for any point's neighborhood) there are no $M > 0, 0 < a < 1$ such that $|\xv_{mp}| \leq Ma^m|\xv_0|$ which implies that the PETC system is not GES. This is a contradiction.
\end{proof}

Proposition \ref{prop:petchurwitz} implies that stable fixed or periodic sampling patterns generated by a PETC system can be used in a multi-rate periodically sampled system, which will also render the origin GES. Note that the existence of such a stable periodic sampling pattern does not imply that the PETC generates that pattern everywhere; as a matter of fact, it may generate sequences that converge to this stable sequence. In these cases, the PETC has a rival periodic sampling schedule which also achieves GES. %
\footnote{While both approaches stabilize the system with equal limit average sampling performances, their transients should be different. It remains to be investigated if their asymptotic performance properties, i.e., GES decay rates, are the same.}
This is not necessarily true if no stable periodic pattern is exhibited, i.e., when PETC exhibits chaotic or aperiodic traffic.

\begin{rem}\label{rem:CETCnotSchur}
Proposition \ref{prop:petchurwitz} and its associated conclusion are not true for CETC. For example, consider the case 2 from Ex.~\ref{ex:r2}: its stable fixed point occurs for the inter-event time $y \approx 0.3903$; the eigenvalues of $\Mm(y)$ are $0.757$ (which is $1/(1+\sigma)$ as expected from Prop.~\ref{prop:fixedraytabuada}) and $-1.33$,  hence $\Mm(y)$ is not Schur. Given Prop.~\ref{prop:petchurwitz}, it is now not surprising that case 2's PETC implementation (Fig.~\ref{fig:ex1c5}) does not exhibit an asymptotically stable inter-event time trajectory. More interestingly, this stays true regardless of how small $h$ is.
\end{rem}

This Section has presented many properties of fixed and periodic subsets of ETC, such as dimensional conditions for fixed o-lines and o-planes to exist, how to find them, and how to characterize their attractivity. However, it has not yet provided a means to compute the limit metrics or their robust versions. Looking again at Example \ref{ex:r2}, it is clear that several challenges remain: 
\begin{enumerate}
\item If a stable fixed or periodic pattern is found, can we ensure that it is almost globally attractive? (Here, almost is used to exclude the finitely many unstable fixed or periodic patterns, in case these exist.)
\item If $f$ has fixed or periodic patterns, how can we obtain some information about the limit metrics?
\item If multiple fixed or periodic patterns are found, but inside a chaotic invariant set, how to compute robust limit metrics?
\end{enumerate}
The next Section provides (partial) answers to these questions for PETC using a symbolic approach. 

\section{Quantitative analysis: a symbolic approach}\label{sec:symbolic}

In this section, we shift from the nonlinear analysis tools used in Sec.~\ref{sec:behaviors} to symbolic tools in the spirit of \cite{tabuada2009verification}. We focus on PETC, whose discrete-output nature facilitates the construction of finite-state models \cite{gleizer2020scalable}. For this part, it is necessary to introduce some formalism and previous results.

\subsection{Transition systems, simulations, and quantitative automata}

In \cite{tabuada2009verification}, Tabuada gives a generalized notion of transition system:
\begin{defn}[Transition System \cite{tabuada2009verification}]\label{def:system} 
A system $\Ss$ is a tuple $(\Xs,\Xs_0,\Es,\Ys,H)$ where:
\begin{itemize}
\item $\Xs$ is the set of states,
\item $\Xs_0 \subseteq \Xs$ is the set of initial states,
\item $\Es \subseteq \Xs \times \Xs$ is the set of edges (or transitions),
\item $\Ys$ is the set of outputs, and
\item $H: \Xs \to \Ys$ is the output map.
\end{itemize}
\end{defn}
Here we have omitted the action set $\Us$ from the original definition because we focus on autonomous systems. A system is said to be finite (infinite) state when the cardinality of $\Xs$ is finite (infinite). A transition in $\Es$ is denoted by a pair $(x, x')$. We define $\Post_\Ss(x) \coloneqq \{x'\mid (x,x') \in \Es\}$ as the set of states that can be reached from $x$ in one step. System $\Ss$ is said to be \emph{non-blocking} if $\forall x \in \Xs, \Post_\Ss(x) \neq \emptyset.$ 
We call $x_0x_1x_2...$ an \emph{infinite internal behavior}, or \emph{run} of $\Ss$ if $x_0 \in \Xs_0$ and $(x_i,x_{i+1}) \in \Es$ for all $i \in \N$, and $y_0y_1...$ its corresponding \emph{infinite external behavior}, or \emph{trace}, if $H(x_i) = y_i$ for all $i \in \N$. We denote by $B_{\Ss}(r)$ the external behavior from a run $r = x_0x_1...$ (in the case above, $B_{\Ss}(r) = y_0y_1...$), by $\Bs^l_x(\Ss)$ (resp.~$\Bs^+_x(\Ss)$ and $\Bs^\omega_x(\Ss)$) the set of all $l$-long (resp.~finite and infinite) external behaviors of $\Ss$ starting from state $x$, and by $\Bs^l(\Ss) \coloneqq \bigcup_{x\in\Xs_0}\Bs^l_x(\Ss)$  (resp.~$\Bs^+(\Ss) \coloneqq \bigcup_{x\in\Xs_0}\Bs^+_x(\Ss)$ and $\Bs^\omega(\Ss) \coloneqq \bigcup_{x\in\Xs_0}\Bs^\omega_x(\Ss)$) the set of all $l$-long (resp.~finite and infinite) external behaviors of $\Ss$. 

The concepts of simulation and bisimulation are fundamental to establish formal relations between two transition systems.

\begin{defn}[Simulation Relation \cite{tabuada2009verification}]\label{def:sim}
Consider two systems $\Ss_a$ and $\Ss_b$ with $\Ys_a$ = $\Ys_b$. A relation $\Rs \subseteq \Xs_a \times \Xs_b$ is a simulation relation from $\Ss_a$ to $\Ss_b$ if the following conditions are satisfied:
\begin{enumerate}
\item[i)] for every $x_{a0} \in \Xs_{a0}$, there exists $x_{b0} \in \Xs_{b0}$ with $(x_{a0}, x_{b0}) \in \Rs;$
\item[ii)] for every $(x_a, x_b) \in \Rs, H_a(x_a) = H_b(x_b);$
\item[iii)] for every $(x_a, x_b) \in \Rs,$ we have that $(x_a, x_a') \in \Es_a$ implies the existence of $(x_b, x_b') \in \Es_b$ s.t.~$(x_a', x_b') \in \Rs.$
\end{enumerate}
\end{defn}
We say $\Ss_a \preceq \Ss_b$ when $\Ss_b$ simulates $\Ss_a$, which is true if there exists a simulation relation from $\Ss_a$ to $\Ss_b$. When $\Rs$ is a simulation relation from $\Ss_a$ to $\Ss_b$ and also $\Rs^{-1}$ is from $\Ss_b$ to $\Ss_a$, we say that $\Ss_a$ and $\Ss_b$ are \emph{bisimilar}, and denote by $\Ss_a \bisim \Ss_b$. 
Weaker but relevant relations associated with simulation and bisimulation are, respectively, \emph{behavioral inclusion} and \emph{behavioral equivalence}:
\begin{defn}[Behavioral inclusion and equivalence \cite{tabuada2009verification}]
Consider two systems $\Ss_a$ and $\Ss_b$ with $\Ys_a$ = $\Ys_b$. We say that $\Ss_a$ is \emph{behaviorally included} in $\Ss_b$, denoted by $\Ss_a \preceq_\Bs \Ss_b$, if $\Bs^\omega(\Ss_a) \subseteq \Bs^\omega(\Ss_b).$ In case $\Bs^\omega(\Ss_a) = \Bs^\omega(\Ss_b),$ we say that $\Ss_a$ and $\Ss_b$ are \emph{behaviorally equivalent}, which is denoted by $\Ss_a \bisim_\Bs \Ss_b$.
\end{defn}
(Bi)simulations lead to behavioral inclusion (equivalence):
\begin{thm}[\!\!\cite{tabuada2009verification}]\label{thm:simimpliesbeh}
Given two systems $\Ss_a$ and $\Ss_b,$ 
$\Ss_a \preceq \Ss_b \implies \Ss_a \preceq_\Bs \Ss_b$ and $\Ss_a \bisim \Ss_b \implies \Ss_a \bisim_\Bs \Ss_b$.
\end{thm}

If $\Ss$ is finite-state, we can associate a digraph $G$ with it, where states are nodes and an edge $x \to x'$ exists if $(x,x') \in \Es$. A digraph has an associated $(0,1)$-matrix, the incidence matrix $\Tm$, obtained by attributing an index $i$ to each node; then $T_{ij} = 1$ if $x_i \to x_j$, $T_{ij} = 0$ otherwise. We say that $\Tm$ is the incidence matrix of $\Ss$.

For quantitative analysis of system properties, we resort to the framework of \cite{chatterjee2010quantitative}, with the adaptations made in \cite{gleizer2021hscc} to include output maps.
\begin{defn}[Weighted transition system \cite{gleizer2021hscc}]
A weighted transition system (WTS) $\Ss$ is the tuple $(\Xs,\Xs_0,\Es,\Ys,H, \gamma)$, where
\begin{itemize}
\item $(\Xs,\Xs_0,\Es,\Ys,H)$ is a \emph{non-blocking} transition system;
\item $\gamma: \Es \to \Q$ is the \emph{weight function}.
\end{itemize}
\end{defn}
For a given run $r = x_0x_1...$ of $\Ss$, abusing notation, $\gamma(r) = v_0v_1...$ is the sequence of weights defined by $v_i = \gamma(x_i,x_{i+1})$. We use $\gamma_i(r)$ for the $i$-th element of $\gamma(r)$. A WTS is called \emph{simple} if for all $(x,x') \in \Es, \gamma(x,x') = H(x)$ \cite{gleizer2021computing}; in this case $\gamma(r) = B_{\Ss}(r)$, i.e., the set of weight sequences of $\Ss$ is equal to its behavior. All WTSs we consider in this work are simple, so hereafter we focus on this case.  In this case, we define the following \emph{values} of a behavior set $\Bs \subseteq 2^{\N\to\Q}$, %
in the spirit of the metrics presented in §\ref{ssec:prob}:

\begin{align*}
\InfLimInf(\Bs) &\coloneqq \inf\left\{\liminf_{i\to\infty}y_i \;\middle|\; \{y_i\} \in \Bs \right\}, \\ 
\InfLimAvg(\Bs) &\coloneqq \inf\;\biggl\{\liminf_{n\to\infty}\frac{1}{n+1}\sum_{i=0}^{n}y_i \;\bigg|\; \{y_i\} \in \Bs\biggr\}.
\end{align*}
For a value $V \in \{\InfLimInf, \InfLimAvg\}$, we often use the shorthand notation $V(\Ss) \coloneqq V(\Bs^\omega(\Ss))$. 
The following result is extracted from \cite[Theorem 3]{chatterjee2010quantitative} and its proof:
\begin{thm}\label{thm:limavg}
Given a finite-state WTS $\Ss$,
\begin{enumerate}
\item $\InfLimInf(\Ss)$ can be computed in $\bigO(|\Xs|+|\Es|)$; moreover, there exists $x \in \Xs$ such that $H(x) = \InfLimInf(\Ss)$ and $x$ belongs to a strongly connected component (SSC) of the graph defined by $\Ss$.
\item $\InfLimAvg(\Ss)$ can be computed in $\bigO(|\Xs||\Es|)$.  Moreover, system $\Ss$ admits a cycle $x_0x_1...x_k$ satisfying $x_i \to x_{i+1}, i < k,$ and $x_k \to x_0$, s.t. the run $r = (x_0x_1...x_k)^\omega$ satisfies $\liminf_{n\to\infty}\frac{1}{n+1}\sum_{i=0}^{n}\gamma_i(r)) = \InfLimAvg(\Ss)$. 
\end{enumerate}
\end{thm}
The importance of Theorem \ref{thm:limavg} for this work is that global values of limit metrics are computable for finite-state systems. This is fundamentally different from the infinite-state case, where a qualitative analysis is possible, but it is extremely challenging to determine regions of attraction of fixed lines, 
or computing tight bounds on the metrics when no fixed (periodic) solutions are found.
\begin{rem}\label{rem:algo} The algorithm for computing $\InfLimInf(\Ss)$ and $\SupLimSup(\Ss)$ is the same as the one to determine B\"uchi acceptance, and consists of computing SSCs and performing reachability to those \cite{chatterjee2010quantitative}.
The cycle mentioned in Theorem \ref{thm:limavg} is a \emph{minimum average cycle} (MAC) of the weighted digraph defined by $\Ss$. The algorithm to compute the value is due to \cite{karp1978characterization}, which also detects reachable SCCs and employs dynamic programming on those. The cycle can be recovered in $\bigO(|\Xs|)$ using the algorithm in \cite{chaturvedi2017note}.
\end{rem}
\begin{rem}\label{rem:algohscc}
In case $\Ss$ is infinite, a method to compute $\InfLimAvg(\Ss)$ using abstractions was proposed in \cite{gleizer2021hscc, gleizer2021computing}, and the same results can be extended to $\InfLimInf$. The main idea is to compute the metric on the abstraction and retrieve a cycle $\sigma$ that attains the minimum value (a MAC when computing ILA or any cycle in the SCC that attains the ILI). The value of the abstraction is a lower bound to the value of the concrete system \cite{gleizer2021computing}. Then, one verifies if $\sigma^\omega \in \Bs^\omega(\Ss)$ (in the PETC case, by using Theorem \ref{thm:verifycycle} with Prop.~\ref{prop:quadsubspace}): if true, then the value of the abstraction is in fact \emph{equal} to the value of the concrete system \cite{gleizer2021hscc}; if not, one can refine the abstraction and reiterate. The next subsection presents how to abstract a PETC traffic model and refine it.
\end{rem}

\subsection{\ensuremath{\mathit{l}}-Complete PETC traffic models}

Here we recover results of our previous work \cite{gleizer2020towards, gleizer2021hscc}, which determines how to build a \emph{finite-state system} that captures sequences of $l$ inter-sample times from system \eqref{eq:samplemap} with $\Ts = h\N$ (PETC) and their associated state-space partition. First, let us describe the system \eqref{eq:samplemap} as a transition system:
\begin{equation}\label{eq:S}
\begin{aligned}
\Ss = (\R^\nx&, \R^\nx, \Es, \Ys, H, \gamma) , \text{ where} \\ 
\Es & = \{(\xv,\xv') \in \R^n \times \R^n \mid \xv' = \Mm(\tau(\xv))\xv\} \\
\Ys &= \{h,2h,...,\bar{\tau}\} \\ 
H(\xv) &= \gamma(\xv,\xv') = \tau(\xv).
\end{aligned}
\end{equation}
Denote by $\Ks \coloneqq \Ys/h$, the set of possible inter-event times normalized by $h$.

\begin{defn}[\!\!\cite{gleizer2021hscc}]\label{def:lsim} Given an integer $l \geq 1$, the \emph{$l$-complete PETC traffic model} of system $\Ss$ from \eqref{eq:S} is the system $\Ss_l \coloneqq \left(\Xs_l, \Xs_l, \Es_l, \Ys, H_l, \gamma_l \right)$, with 
\begin{itemize}
\item $\Xs_l \coloneqq \Bs^l(\Ss)$,
\item $\Es_l = \{(k\sigma, \sigma k') \mid k,k' \in \Ks, \sigma \in \Ks^{l-1}, k\sigma, \sigma k' \in \Xs_l\},$
\item $H_l(k_1k_2...k_m) = \gamma_l(k_1k_2...k_m, \cdot) = hk_1.$
\end{itemize}
\end{defn}

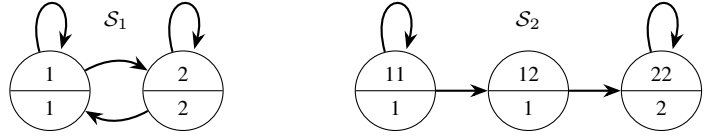
\begin{figure}
\centering
\begin{tikzpicture} [node distance = 1.75cm, on grid, auto, font=\footnotesize]
	\node (1) [state with output, minimum size=3em] {1 \nodepart{lower} 1};
	\node (2) [state with output, right = of 1, minimum size=3em] {
		2 \nodepart{lower} 2    };
	\path [-Stealth, thick]
	(1) edge[bend left]  (2) (2) edge[bend left] (1)
	(1) edge [loop above, >=Stealth] (1)
	(2) edge [loop above, >=Stealth] (2);
	
	\path[draw=none] (1) -- coordinate[midway](m) (2);
	\node (s1) [above =0.7cm of m] {$\Ss_1$};
\end{tikzpicture}~~~~~~~~~~~~
\begin{tikzpicture} [node distance = 1.75cm, on grid, auto, font=\footnotesize]
	\node (11) [state with output, minimum size=3em] {11 \nodepart{lower} 1};
	\node (12) [state with output, right = of 11, minimum size=3em] {
		12 \nodepart{lower} 1    };
	\node (22) [state with output, right = of 12, minimum size=3em] {
		22 \nodepart{lower} 2    };
	\path [-Stealth, thick]
	(11) edge  (12) (12) edge (22)
	(11) edge [loop above, >=Stealth] (11)
	(22) edge [loop above, >=Stealth] (22);
	
	\node (s1) [above =0.7cm of 12.center] {$\Ss_2$};
\end{tikzpicture}
\caption{\label{fig:lcomplete} $l$-complete models of the illustrative PETC system of Fig.~\ref{fig:modesandcones}, for $l=1$ (left) and $l=2$ (right). Each node represents a state, with the top label being the state label and the bottom being its output.}
\end{figure}

The state space of the model above is the set of $l$-long outputs that the PETC system $\Ss$ can generate, which can be computed by using the techniques described in \cite{gleizer2021hscc}. 
The output of a state $x$ is its next inter-sample time (divided by $h$), which is also the weight of any transition leaving $x$. The transition relation is what is called in \cite{schmuck2015comparing} the domino rule: a state associated with a sequence $k_1k_2...k_l$ must lead to a state whose next first $l-1$ samples are $k_2k_3...k_l$, because the system is deterministic, autonomous, and time-invariant. Hence, any state in $\Xs_l$ that starts with $k_2k_3...k_l$ is a possible successor of $k_1k_2...k_l$. Note that both $\Ss$ and $\Ss_l$ are simple WTSs. %
The following result gives the desired simulation refinement properties:
\begin{prop}[\!\!\cite{gleizer2021hscc}]\label{prop:refinements}
Consider the system $\Ss$ from Eq.~\eqref{eq:S} and $\Ss_l$ from Def.~\ref{def:lsim}, for some $l \geq 1$. Then, $\Ss \preceq \Ss_{l+1} \preceq \Ss_l,$ which implies that $\Ss \preceq_\Bs \Ss_{l+1} \preceq_\Bs \Ss_l.$
\end{prop}
Fig.~\ref{fig:lcomplete} shows $l$-complete models $\Ss_1$ and $\Ss_2$ for the illustrative PETC example of Fig.~\ref{fig:modesandcones}: note that $(1,2) \in \Xs_2$ means that there are points in $\R^3$ that belong to $\Qs_1$, but the next sample would belong to $\Qs_2$; at the same time, $(2,1) \notin \Xs_2$, which implies that no points leave $\Qs_2$ after sampling, i.e., $\Qs_2$ is forward-invariant. This kind of observation is central when using abstractions to differentiate robust from fragile behaviors.

\subsection{Robust limit metrics}\label{ssec:robmet}

Limit metrics of PETC traffic can be computed using abstractions as described in Remark \ref{rem:algohscc}. However, as discussed in §\ref{ssec:prob}, these metrics can be rare in the sense that they only occur from a zero-measure initial set. E.g., revisiting cases 2 and 3 of Example \ref{ex:r2}, we have 1 and 2 unstable fixed points, respectively. In both cases, the unstable fixed point near $\theta=-1.3$ gives the value of InfLimInf and InfLimAvg; but for all other initial conditions $\theta_0$, trajectories $\{\theta_i\}$ are attracted to the stable fixed point in case 2 and the stable period-4 orbit in case 3. Thus, \emph{robust limit metrics should be oblivious to unstable orbits.} Let us properly define what stable and unstable behaviors are for systems with a finite output set:%
\begin{defn}[Stable behaviors]\label{def:stablebehaviors}
Consider a deterministic WTS $\Ss$ where $\Xs$ is a metric space and $\Ys$ is finite. A periodic behavior $\sigma^\omega \in \Bs^\omega(\Ss)$ is said to be \emph{stable} if there exists $x \in \Xs$ with a neighborhood $\Us$ such that every $x' \in \Us$ satisfies $\Bs^\omega_x = \Bs^\omega_{x'} = \sigma^\omega$ (as in Def.~\ref{def:stablepetc}). 
\end{defn}
\begin{defn}[Robust limit metrics]\label{def:robustmetrics}
Let $\Ss$ be a simple WTS, $\Bs_{\mathrm{u}}^\omega(\Ss)$ be the set of its unstable behaviors, and $V$ be a system limit metric ($\InfLimAvg$ or $\InfLimInf$). Then the robust version of the metric is
$ \Rob V(\Ss) \coloneqq V(\Bs^\omega(\Ss) \setminus \Bs_{\mathrm{u}}^\omega(\Ss)). $
\end{defn}
Removing unstable behaviors as the ones discussed above is safe in that small perturbations in the initial state lead to distant behaviors. However, consider case 4 of Example \ref{ex:r2} and its chaotic invariant set: it has infinitely many unstable orbits, and almost every orbit comes arbitrarily close to those orbits. In fact, due to transitivity, \emph{every initial solution starting on the chaotic invariant set will come arbitrarily close to any unstable orbit within it.} Thus, the infimum of a set of metrics on behaviors on a chaotic set, even when excluding the unstable ones, can be equal to one of its unstable behaviors. This deserves a further distinction between unstable behaviors.
\begin{defn}[Absolutely unstable behaviors]\label{def:aubehavior}
Consider a deterministic WTS $\Ss$ where $\Xs$ is a metric space and $\Ys$ is finite. A periodic behavior $\sigma^\omega \in \Bs^\omega(\Ss)$ is said to be \emph{absolutely unstable} (a.u.) if it is unstable and for almost all $x$ there exists $L \in \N$ such that $\forall l > L, $ $\sigma^l$ is not a subsequence of $\Bs_x^\omega(\Ss)$. The set of a.u.~behaviors of $\Ss$ is denoted by $\Bs_{\mathrm{au}}^\omega(\Ss).$
\end{defn}
A.u. behaviors are fragile in the sense that small perturbations to initial states lead to substantially different behaviors. 

Periodic behaviors of a PETC system $\Ss$ that occur in an abstraction $\Ss_l$ can be verified to be (absolutely) unstable (see proof in the Appendix).
\begin{prop}\label{prop:unstable}
Consider system $\Ss$ from Eq.~\eqref{eq:S} and let $\sigma^\omega \in \Bs^\omega(\Ss)$. Assume $\Mm_k$ is nonsingular for all $k \in \{1,...,\bar{k}\}$. Further, assume $\Mm_\sigma$ is mixed, and let $\vv_1, \vv_2, ..., \vv_n$ be the unitary eigenvectors of $\Mm$ ordered from largest-in-magnitude corresponding eigenvalue to smallest. Denote by $\As$ any linear invariant of $\Mm_\sigma$ containing $\vv_1$. (I) If $\As \nsubseteq \cl(\Qs_\sigma)$, then $\sigma^\omega$ is an unstable behavior. (II) If additionally the cycle $x_1x_2...x_c$ in $\Ss_l$ that generates $\sigma^\omega$ (i.e., $B_{\Ss_l}(\{x_1x_2...x_c\}^\omega) = \sigma^\omega$) is the only cycle of its SCC, then $\sigma^\omega$ is absolutely unstable in $\Ss$.
\end{prop}
Hereafter we shall denote a linear invariant $\As$ containing $\vv_1$ as in Prop.~\ref{prop:unstable} a \emph{dominant} linear invariant, after the concept of dominant modes in linear systems. %
Referring again to Fig.~\ref{fig:modesandcones} and the corresponding 2-complete model (Fig.~\ref{fig:lcomplete}), we see two periodic behaviors, $1^\omega$ and $2^\omega$. The illustrated o-line is an invariant of $\Mm(1)$ that is not dominant (as can be inferred by the trajectory of gray points that diverge from the line); moreover, the cycle of $\Ss_2$ that generates $1^\omega$ is a simple cycle, the node 11 with a self loop. This implies that $1^\omega$ is absolutely unstable. Note that this conclusion could not be obtained by inspecting $\Ss_1$, which is a complete graph without simple cycles. The behavior $2^\omega$, on the other hand, is stable.

Clearly, removing only a.u.~behaviors is safe to give a lower bound estimate to $\Rob V(\Ss)$, i.e., $V(\Bs^\omega(\Ss) \setminus \Bs_{\mathrm{au}}^\omega(\Ss)) \leq V(\Bs^\omega(\Ss) \setminus \Bs_{\mathrm{u}}^\omega(\Ss))$. An equality holds when $\Ss$ is not chaotic, since all unstable behaviors are also absolutely unstable. Therefore, determining when $\Ss$ is or is not chaotic is critical to compute the exact value of $\Rob V(\Ss)$. As we see next, chaos on $\Ss$ can be estimated from the abstraction $\Ss_l$.

\subsection{Estimating chaos in abstractions}

In this section we show how to detect (and quantify) chaos on a PETC traffic model $\Ss$, and when one can conclude that $\Ss$ is not chaotic. A commonly used measure of chaos is the topological entropy $h(\Ss)$ \cite{robinson1999dynamical}, satisfying $h(\Ss) \geq 0$, with $h(\Ss) = 0$ implying there is no chaos. However, instead of a topological measure, we are interested in a measure of chaos of the output of the system: if the state is behaving chaotically but this is not reflected in the output, it does not interfere in the metrics we are interested. Therefore, we shall introduce here a notion called \emph{behavioral entropy}, which is a natural extension of the original concept.%

\begin{defn}[Behavioral entropy]\label{def:entropy} Consider a system $\Ss$ and equip $\Ys$ with a metric $d$. A set $\Ws \subset \Bs^\omega(\Ss)$ is called $(n,\epsilon)$-separated if for all behaviors $\yv, \yv'\in\Ws$, where $\yv = y_0y_1...y_i...$ and $\yv' = y'_0y'_1...y'_i...,$ we have $d(y_i,y_i') > \epsilon$ for all $i \leq n$. Let $s(n,\epsilon,\Ss)$ be the maximum cardinality of any $(n,\epsilon)$-separated set. The behavioral entropy is the quantity
\begin{equation}\label{eq:entropy}
h(\Ss) \coloneqq \lim_{\epsilon\to 0}\limsup_{n\to\infty}\frac{\log(s(n,\epsilon,\Ss))}{n}.
\end{equation}
In particular, if $|\Ys| < \infty$ and the distance metric is $d(y,y') = 0$ if $y = y'$ and $d(y,y') = 1$ otherwise, we can ignore the $\epsilon$ component, and it turns out that
\begin{equation}\label{eq:wordentropy}
h(\Ss) = \limsup_{n\to\infty}\frac{\log(N(n,\Ss))}{n},
\end{equation}
where $N(n,\Ss)$ is the number of different words of length $n$ over the alphabet $\Ys$ that are possible trace segments of $\Ss$.

A system is called \emph{behaviorally chaotic} whenever its behavioral entropy is positive. 
\end{defn}
\begin{rem}\label{rem:entropy_subshift}
The topological entropy also takes the form in Eq.~\eqref{eq:wordentropy} for \emph{subshifts of finite type}, an abstraction used for autonomous dynamical systems to study their topological properties (see \cite{robinson1999dynamical}). 
\end{rem}

Definition \ref{def:entropy} takes a behavioral approach \cite{willems1991paradigms} to extend the original definition \cite{robinson1999dynamical} for systems that are possibly nondeterministic and have output maps. If $H = \Id$ and $\Post(x) = \{f(x)\}$ for some continuous map $f : \Xs \to \Xs$, we recover the original notion. It may seem unproductive to extend a measure of chaos to non-deterministic systems, as these should all be chaotic in some sense; however, this is not always the case. %
For example, consider $\Ss_2$ of Fig.~\ref{fig:lcomplete}: %
it is easy to see that $N(n,\Ss_2) = n+1:$ $11...11, 11...12, ..., 122...22, 22...22$. Hence, $h(\Ss_2) = \lim_{n\to\infty}(\log(n+1)/n) = 0,$ and this system is not (behaviorally) chaotic.

\begin{prop}\label{prop:entropybounds}
Consider two transition systems $\Ss_a$ and $\Ss_b$ with $\Ys_a = \Ys_b = \Ys$ s.t.~$\Ss_a \preceq_{\Bs} \Ss_b$. If $|\Ys| < \infty$, then $h(\Ss_a) \leq h(\Ss_b)$.
\end{prop}

\begin{proof}
Trivially, from behavioral inclusion \cite{tabuada2009verification}, $\forall n \in \N, N(n,\Ss_a) \leq N(n,\Ss_b)$. The result follows from monotonicity of the $\log$ function.
\end{proof}

The question now is how to compute the behavioral entropy of a finite-state system. This result is known for topological entropy of subshifts of finite type, which is the same as a finite-state transition system with $H = \Id$:

\begin{thm}[\!\!{{\cite[Theorem IX.1.9]{robinson1999dynamical}}}]\label{thm:subshiftentropy}\!\!%
\footnote{In \cite{robinson1999dynamical}, the internal behavior from an initial state is called \emph{itinerary}. The original Theorem states that this quantity is also the topological entropy of the subshift $\Ss$, but here we only need the formula relating the limit to the spectral radius of $\Tm$.}
Let $\Ss$ be a finite system and $N'(n,\Ss)$ be the number of different $n$-length words over the alphabet $\Xs$ generated by $\Ss$ (note that this reflects the \emph{internal} behavior of $\Ss$). Then,
\begin{equation*}
\limsup_{n\to\infty}\frac{\log(N'(n,\Ss))}{n} = \log\lambda_1(\Tm),
\end{equation*}
where $\Tm$ is the incidence matrix of $\Ss$.
\end{thm}

Under a detectability condition of $\Ss$, the same result holds for behavioral entropy:

\begin{defn}[Detectability]
A transition system $\Ss$ is said to be $l$-detectable if there exists a finite $l \in \N$ such that, for each word $w \in \Bs^+(\Ss), |w| \geq l$, there exists a unique $x \in \Xs$ such that $w \in \Bs^+_x(\Ss)$.
\end{defn}

\begin{thm}\label{thm:valueofentropy}
Let $\Ss$ be an $l$-detectable finite-state system for some finite $l \in \N$, and let $\Tm$ be its incidence matrix. Then,
\begin{equation}\label{eq:behavioralentropy}
h(\Ss) = \log\lambda_1(\Tm).
\end{equation}
\end{thm}

\begin{proof}
Let $s \coloneqq |\Ys|$. Because of $l$-detectability, every $(n+l)$-long external behavior of $\Ss$ gives a unique $n$ internal behavior, hence $N(n+l,\Ss) \geq N'(n,\Ss)$. From every  external behavior of length $n$, there can be at most $s^l$ external behaviors of length $n+l$ (simply concatenate every possible word in $\Ys^l$ to complete the length). Thus, $s^lN(n,\Ss) \geq N(n+l,\Ss)$.
Finally, since the output map $H$ is single-valued, the number of external behaviors can never be bigger than the number of different internal behaviors: $N(n,\Ss) \leq N'(n,\Ss)$. Combining these inequalities, the following holds for all $n > l$:
$$ N(n,\Ss) \leq N'(n,\Ss) \leq s^lN(n,\Ss). $$
Now, 
\begin{equation*} 
\limsup_{n\to\infty}\frac{\log(s^lN(n,\Ss))}{n}  = \limsup_{n\to\infty}\left(\frac{\log(s^l)}{n} + \frac{\log(N(n,\Ss))}{n}\right) 
= \limsup_{n\to\infty}\frac{\log(N(n,\Ss))}{n}.
\end{equation*}
The sandwich rule and Theorem \ref{thm:subshiftentropy} conclude the proof.
\end{proof}

The following results help us apply Theorem \ref{thm:valueofentropy} to the PETC traffic model.

\begin{prop}\label{prop:simplecycles}
A non-blocking finite-state $l$-detectable autonomous transition system $\Ss$ has zero behavioral entropy if and only if all the strongly connected components (SCCs) of its associated graph are isolated nodes or simple cycles.
\end{prop}
\begin{proof}
The spectrum of a digraph is the union of the spectra of its SCCs \cite{brualdi2010spectra}. Because $\Ss$ is non-blocking, it must have at least one cycle. The adjacency matrix of an isolated node is $[\,0\,]$, thus its spectrum is $\{0\}$. Further, all vertices of a simple cycle have only one outgoing edge, hence the corresponding SCC has a constant outdegree of 1. From \cite[Theorem 2.1]{brualdi2010spectra}, the spectral radius of an SCC is 1 iff it has constant outdegree 1. Hence, the spectral radius of the whole graph is $\max(1,0) = 1$, whose log is 0.
\end{proof}

\begin{rem}\label{rem:lcompleteisldetectable}
The $l$-complete PETC traffic model of Def.~\ref{def:lsim} is $l$-detectable because, by definition, each $k_1k_2...k_l \in \Xs_l$ is the unique state that generates the finite behavior $hk_1, hk_2, ... hk_l$.
\end{rem}

\begin{thm}\label{thm:petcentropy}
Consider the PETC system \eqref{eq:plant}--\eqref{eq:quadtrig} ($\Ts=h\N$), its traffic model $\Ss$ from Eq.~\eqref{eq:S} and its $l$-complete traffic model (Def.~\ref{def:lsim}) $\Ss_l,$ with $l \in \N$. The following assertions are true:
\begin{enumerate}[i)]
\item $h(\Ss) \leq h(\Ss_l)$;
\item If all SCCs of $\Ss_l$ are simple cycles, then $h(\Ss) = h(\Ss_l) = 0$, i.e., $\Ss$ is not chaotic. 
\end{enumerate}
\end{thm}

\begin{proof}
Assertion (i): Prop.~\ref{prop:refinements} gives that $\Ss \preceq \Ss_l$; then, from Theorem \ref{thm:simimpliesbeh}, $\Ss \preceq_\Bs \Ss_l$; finally, Prop.~\ref{prop:entropybounds} concludes the proof.
Assertion (ii): $\Ss_l$ satisfies the premises of Prop.~\ref{prop:simplecycles}. Hence, $h(\Ss_l) = 0$. Using assertion (i) and the fact that $h(\Ss) \geq 0$, we conclude that $h(\Ss) = 0$.
\end{proof}

Revisiting Fig.~\ref{fig:lcomplete}, it is easy to see that $h(\Ss_1) = \log(2) = 1$ bit (base 2), while $h(\Ss_2) = 0$, which implies that the example of Fig.~\ref{fig:modesandcones} is not chaotic.

\subsection{Estimating and computing robust metrics}

Now we are equipped with the necessary tools to estimate robust limit metrics using an abstraction and determine when they are equal to the concrete system's or simply a lower bound. Based on the discussion in §\ref{ssec:robmet}, we define the following robust limit metric for the abstraction:
\begin{defn}[Robust metric for $\Ss_l$] \label{def:robustsl}
Consider system $\Ss$ from Eq.~\eqref{eq:S} and an $l$-complete model for it, $\Ss_l$ (Def.~\ref{def:lsim}). Let $\tilde\Bs_{\mathrm{au}}^\omega(\Ss_l)$ be the set of behaviors of $\Ss_l$ that are are simple cycles in $\Ss_l$ and are absolutely unstable in $\Ss$. We define $\Rob V(\Ss_l)$ as $V(\Bs^\omega(\Ss_l) \setminus \tilde\Bs_{\mathrm{au}}^\omega(\Ss_l))$.
\end{defn}

\begin{thm}\label{thm:robustvalue}
Consider system $\Ss$ from Eq.~\eqref{eq:S} and its $l$-complete model $\Ss_l$. Consider $V \in \{\InfLimInf, \InfLimAvg\}$; then $\Rob V(\Ss'_l) \leq \Rob V(\Ss)$. Moreover, if all SCCs of $\Ss_l$ are simple cycles, and the minimizing cycle $\sigma$ satisfies $\sigma^\omega \in \Bs^\omega(\Ss)$, 
then $\Rob V(\Ss_l) = \Rob V(\Ss).$ 
\end{thm}
\begin{proof}
Because all behaviors in $\tilde{\Bs}_{\mathrm{au}}$ are absolutely unstable in $\Ss$, we have $\tilde\Bs_{\mathrm{au}}^\omega(\Ss_l) \subseteq \Bs_{\mathrm{au}}^\omega(\Ss)$, and thus $\tilde\Bs_{\mathrm{au}}^\omega(\Ss_l) \subseteq \Bs_{\mathrm{u}}^\omega(\Ss)$. From Prop.~\ref{prop:refinements}, $\Bs^\omega(\Ss_l) \supseteq \Bs^\omega(\Ss)$; hence $\Bs^\omega(\Ss_l) \setminus \tilde{\Bs}_{\mathrm{au}}^\omega(\Ss_l) \supseteq  \Bs^\omega(\Ss) \setminus \Bs_{\mathrm{u}}^\omega(\Ss)$. 
Now, for any behavior set $\Bs,$ $V(\Bs) = \inf\{f(y_i) \mid \{y_i\} \in \Bs\} = \inf\{F(\{y_i\}) \mid \{y_i\} \in \Bs\},$ where $F(\{y_i\})$ is either $\liminf_{i\to\infty}y_i$ (ILI) or $\liminf_{n\to\infty}\frac{1}{n+1}\sum_{i=0}^{n}y_i$ (ILA). Hence, $\Bs_a \subseteq \Bs_b$ implies $V(\Bs_a) \geq V(\Bs_b)$, and the inequality $\Rob V(\Ss) \geq \Rob V(\Ss_l)$ follows.

For the equality: if $\Ss_l$ contains only simple cycles, then $\Ss$ is not behaviorally chaotic (Theorem \ref{thm:petcentropy}), and thus $\Bs_{\mathrm{u}}^\omega(\Ss) = \Bs_{\mathrm{au}}^\omega(\Ss)$ (all unstable cycles are absolutely unstable). Then, the minimizing cycle $\sigma$ of $\Ss_l$ is by exclusion a stable cycle of $\Ss$. 
Since $\sigma^\omega \in \Bs^\omega(\Ss) \setminus \Bs_{\mathrm{u}}^\omega(\Ss),$ we have that $\Rob V(\Ss_l) = F(\sigma^\omega) \geq \inf\{F(\{y_i\}) \mid \{y_i\} \in \Bs^\omega(\Ss) \setminus \Bs_{\mathrm{u}}^\omega(\Ss)\}\ = \Rob V(\Ss)$.
Hence, %
$\Rob V(\Ss_l) = \Rob V(\Ss)$.
\end{proof}

Revisiting Figs.~\ref{fig:modesandcones} and \ref{fig:lcomplete} one last time, we have trivially that $\InfLimInf(\Ss) = \InfLimAvg(\Ss) = 1$, but using Theorem \ref{thm:robustvalue} on $\Ss_2$ we conclude that $\Rob\InfLimInf(\Ss) = \Rob\InfLimAvg(\Ss) = 2$. Nevertheless, by Prop.~\ref{prop:petchurwitz}, $\Mm(2)$ must be Schur, and hence a periodic sampling of $2h$ would also stabilize the system with the same traffic performance.

\begin{rem}\label{rem:ilichaos}
In the case of $\Rob\InfLimInf$, if the invariant associated to the minimizing cycle $\sigma$ can be verified to belong to a chaotic invariant set, under mild assumptions it holds that $\Rob\InfLimInf(\Ss_l) = \Rob\InfLimInf(\Ss)$. To see this, first note that $\Rob\InfLimInf = \min(\sigma) \eqqcolon y$; denoting by $\Xs_{\mathrm{c}}$ the chaotic invariant set, if $\Qs_y \cap \Xs_{\mathrm{c}}$ has non-empty interior, by the Birkhoff Transitivity Theorem (see Def.~\ref{def:trans}) almost every solution starting in $\Xs_{\mathrm{c}}$ visits $\Qs_y$ infinitely often.
\end{rem}

\begin{rem}\label{rem:ergocalc}
In case a chaotic invariant set is ergodic, the infinimal limit average is the same almost everywhere (when restricted to the set), i.e., it is independent of the initial condition (as a consequence of Birkhoff Ergodic Theorem). As a matter of fact, almost everywhere means everywhere except the union of periodic orbits. Thus, $\Rob\InfLimAvg(\Ss_l)$ can then be a conservative estimate. Nevertheless, the associated $\Rob\InfLimAvg$ can be estimated through simulations. Ergodicity can be statistically tested using the approach of \cite{domowitz1993consistent}, where one tests whether the two initially different distributions on $\Xs$ converge to an equal one upon the repeated application of the map $f$ by using a non-parametric hypothesis test such as the Kolmogorov--Smirnov (KS) test. Alternatively, the test can be performed on the distributions of outputs; because $\Ys$ is discrete, an hypothesis test appropriate for discrete supports, such as the Cram\'er--von Mises (CvM) test \cite{arnold2011nonparametric}. For this approach to succeed, it is important that the initial distribution contains only points that are in or lead to the chaotic invariant. The abstraction $\Ss_l$ can be used as an approximate selector of points on the chaotic invariant set, as its SCCs that are not simple cycles are related to over-approximations of potential chaotic invariants on the concrete system $\Ss$.
\end{rem}


\section{Numerical examples}\label{sec:numerical}

We have implemented a program in Python to compute $\Rob\InfLimAvg(\Ss)$ using Theorem \ref{thm:robustvalue}, as well as the non-robust version $\InfLimAvg(\Ss)$ from \cite{gleizer2021hscc}. The program relies in \texttt{SciPy} for linear algebra computations, \texttt{Z3} \cite{demoura2008z3} for computing the state set of the $l$-complete models $\Ss_l$, and \texttt{graph-tool}\cite{peixoto2014graphtool} for efficient graph manipulation. 

\begin{example}\label{ex:num}
Consider system \eqref{eq:plant}--\eqref{eq:quadtrig} with
\begin{equation}
\begin{gathered}
	\Am = \begin{bmatrix}0 & 1 \\ -2 & 3\end{bmatrix}, \ \Bm = \begin{bmatrix}0 \\ 1\end{bmatrix}, \\
	c(s,\xv,\hat\xv) = |\xv - \hat\xv| > \sigma|\xv|,
\end{gathered} \tag{\ref{eq:example2d} revisited}
\end{equation}
as in Example \ref{ex:r2}. Now we use PETC with $h=0.05$ and check the following cases:

\begin{enumerate}
\item $\Km = \begin{bmatrix}0 & -5\end{bmatrix}, \sigma=0.2$, as in Ex.~\ref{ex:r2} case 1
\item $\Km = \begin{bmatrix}0 & -6\end{bmatrix}, \sigma=0.2$.
\item $\Km = \begin{bmatrix}0 & -6\end{bmatrix}, \sigma=0.32$, as in Ex.~\ref{ex:r2} case 2, and Fig.~\ref{fig:ex1c5}.
\end{enumerate}
\end{example}

\begin{table}\caption{\label{tab} ILA values for Example \ref{ex:num}}
	\centering
\begin{tabular}{c|ccc}
\hline
Case & 1 & 2 & 3 \\
\hline
$l$ (robust) & 15 (15) & 10 (10) & 1 (10*) \\
ILA (RobILA) & 0.137 (0.137) & 0.1 (0.25) &  0.1 (0.4) \\
CPU time (robust) [s] & 50 (49) & 23 (19) & 0.81 (5655) \\
\hline
\end{tabular}

\begin{flushleft}
{\footnotesize * Algorithm interrupted before finding a verified cycle.}
\end{flushleft}
\vspace{-1em}
\end{table}

Table \ref{tab} shows the values of ILA and RobILA for each case, as well as the $l$ value at which the algorithms were terminated (or interrupted) and CPU times. Case 1 shows a periodic sequence $\sigma^\omega$ with $|\sigma| = 27$ that is stable and attains both the ILA and the RobILA, as well as $\InfLimInf=\Rob\InfLimInf=0.1$. In fact, case 1 exhibits only this cycle, and a bisimulation is found with $l=27$. Case 2 is different in that an a.u.~cycle is attained at $y=0.1$, but a stable cycle has $y=0.25$ (stationary). Upon inspection of $\Ss_l$, there is another stable cycle at $y=0.3$. 
Unsurprisingly, we also obtain $\InfLimInf(\Ss) = 0.1$ and $\Rob\InfLimInf(\Ss) = 0.25$, which happen at the same cycles. %
Finally, Case 3 is a chaotic example; the ILA is found at $y=\underline{\tau}=0.1$ in the first iteration, but RobILA is never confirmed, although a lower bound of 0.4 is obtained, related to two unstable cycles, $(0.4)^\omega$ and $(0.35, 0.45)^\omega$. However, note the CPU time for obtaining the $\Ss_{10}$ abstraction of approximately 1.5 hour (compare with the others of less than a minute): this is the effect of chaos on the refinements: as indicated by the entropy formula, Eq.~\eqref{eq:entropy}, the number of $l$-sized sequences grows exponentially with $l$. In fact, $\Ss_{10}$ has 9271 states, and an entropy of 1.14 bits. The SCC at which the two cycles belong has 7767 states, a strong indicative of a chaotic invariant set. Figure \ref{fig:entropy} shows the evolution of $h(\Ss_l)$ as a function of $l$ for the three cases, where it is clear that the entropy seems to stabilize at a high value in Case 3, whereas it descends to zero in the other cases. By applying Remark \ref{rem:ergocalc}, two different initial distributions on states related to the large SCC of $\Ss_{10}$ where generated with 1000 points each, and after 9 iterations they converged to the same distribution (CvM test, $p = 0.998$), a good indicative that the chaotic invariant set is ergodic. The average of the obtained ensemble, which by Birkhoff Ergodic Theorem is approximately equal to the limit average of any run starting in the invariant, is 0.417, slightly higher than the 0.4 using Theorem \ref{thm:robustvalue}. %
It is interesting to see that $0.4$ is a slightly higher limit average than what was obtained in the CETC implementation (Ex.~\ref{ex:r2} Case 3, and Remark \ref{rem:CETCnotSchur}), of 0.39; more interestingly, $\Mm(0.4)$ is not Schur, which highlights that the PETC has a larger average sampling period than any stabilizing periodic sampling, at the cost of seemingly unpredictable traffic. Finally, while $\InfLimInf(\Ss) = \InfLimAvg(\Ss) = 0.1$ (at the same unstable cycle $0.1^\omega$) the best lower bound for RobILI is found to be 0.3, which is witnessed by the unstable cycle $(0.3, 0.45, 0.4, 0.5)^\omega$. By inspection, the associated o-line belongs to the chaotic invariant, thus by Remark \ref{rem:ilichaos} this is the correct value of $\Rob\InfLimInf(\Ss)$.

\begin{figure}
\begin{center}
\begin{tikzpicture}

\begin{axis}[
legend cell align={left},
legend style={fill opacity=0.8, draw opacity=1, text opacity=1, draw=white!80!black},
height=4cm,
width=0.6\linewidth,
tick pos=left,
x grid style={white!69.0196078431373!black},
xlabel={\(\displaystyle l\)},
xmajorgrids,
xmin=0.7, xmax=15.3,
xtick style={color=black},
y grid style={white!69.0196078431373!black},
ylabel={\(\displaystyle h(\mathcal{S})\)},
ymajorgrids,
ymin=-0.158496250072112, ymax=3.32842125151443,
ytick style={color=black}
]
\addplot [semithick, red]
table {%
1 2.32192809488736
2 1.44998431347649
3 0.618083653501641
4 0.458879117721529
5 0.411290522974063
6 0.373907348251198
7 0.316144244569658
8 0.293957212000875
9 0.223492743503568
10 0.215663447994219
11 0.208421566032427
12 0.172797770404825
13 0.16755957405296
14 0.162667006825136
15 0.0754903285351724
};
\addlegendentry{Case 1}
\addplot [semithick, blue]
table {%
1 2.8073549220576
2 1.47973405170625
3 0.694241913630619
4 7.72241124409256e-06
5 1.71054532302968e-06
6 6.21868830872191e-06
7 3.36688509673648e-07
8 3.29058460764645e-06
9 2.3277279925263e-06
10 3.8441118045779e-15
};
\addlegendentry{Case 2}
\addplot [semithick, green!50!black]
table {%
1 3.16992500144231
2 1.93329421050452
3 1.30842992114861
4 1.22933473962433
5 1.17498628674993
6 1.16316389183323
7 1.15443894915665
8 1.14942315698291
9 1.14811450243739
10 1.1465629609468
};
\addlegendentry{Case 3}
\end{axis}

\end{tikzpicture}
\caption{\label{fig:entropy} Entropy $h(\Ss_l)$ as a function of $l$ for Example \ref{ex:num}.}
\end{center}
\end{figure}
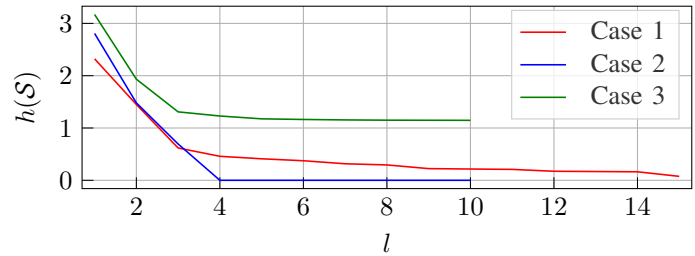

\section{DISCUSSION AND CONCLUSIONS}\label{sec:conclusions}

Event-triggered control can exhibit very complex traffic patterns, and this seems to be more true the more ``aggressive'' the triggering mechanism w.r.t.~sampling reduction. Simple traffic is observed on the opposite case. This is in line with the findings on \cite{postoyan2019interevent} for $\R^2$, in which for small enough triggering parameters the states behave essentially like linear systems: two asymptotes, one stable and one unstable, or a spiral towards the origin when eigenvalues are complex conjugate. This seems to be the case whenever $f$ of the sample system \eqref{eq:samplemap}, projected onto the projective space, is invertible, which is subject of current investigation. 
This would imply that all periodic o-lines or o-planes of the CETC have period one and can be obtained using Remark \ref{rem:cetcfixed}, enabling one to obtain (robust) limit metrics in the continuous case. 

A symbolic method for computing these metrics for PETC was presented in §\ref{sec:symbolic}, and while it is an important first step, it suffers from the curse of dimensionality, particularly when chaotic behaviors are present; future work is aimed at addressing these issues, either by using alternative solvers to Z3, different triggering conditions, or using different abstractions that help pinpoint the existence of chaotic invariant sets. For the latter, an approach such as in \cite{day2008algorithms} may be interesting, which can be seen in the framework of \cite{tabuada2007event} as finding an abstraction that is \emph{backwards} simulated by the concrete system.

We have also seen an example of CETC whose robust infimal limit average is higher than any stable periodic sampling strategy, whereas the same cannot happen with PETC under some generic assumptions. The first case \emph{is a concrete example where ETC is more sampling-efficient than any periodic implementation,} but at the same time any practical implementation of it must rely on periodic checking of triggering times, thus becoming a PETC, and as such any stable inter-event time sequence it exhibits is  stabilizing as a periodic sampling strategy; the only option for PETC to beat the most sampling-efficient periodic implementation would involve chaotic traffic. This is not a problem per se, and one could speculate that chaotic traffic could help in cyber-security aspects, but can make scheduling of multiple ETC loops in a network even more challenging. 

Finally, it is worth noting two important practical observations, one positive and one negative, about this work. The positive one is that it is not actually limited to linear systems: if the closed-loop linear system renders the origin asymptotically stable, and it is linearizable around the origin, then the limit behaviors of the system are those of the linear approximation; hence, our tools can be used to estimate limit metrics of those systems, as was done in \cite{gleizer2021computing}. The negative aspect is that we have considered a very simple case, of state-feedback without disturbances. It is known that doing output-feedback or having disturbances can severely alter the inter-sample behavior of the closed-loop system, in some cases leading to Zeno behavior \cite{borgers2014event}, and practical modifications to the triggering condition are often necessary. It is an open question whether adding these imperfections change our conclusions drastically, or if there are simple adjustments for these cases. Still, it is not difficult to extend the symbolic approach to perturbed systems, following the steps in \cite{delimpaltadakis2020traffic}. Nonetheless, all these conclusions we have obtained for the nominal system cast new light to the long-standing question of how relevant ETC is.

\bibliographystyle{ieeetr} 
\bibliography{mybib}

\appendix

\subsection{Proof of Corollary \ref{cor:lineandplane}}

\begin{proof}
The if and only if statement is a straightforward combination of items (i) and (ii) of Theorem \ref{thm:verifycycle} and assumption (ii) of this corollary.

To see that $\As$ is either an o-line (1-dimensional) or an o-plane (2-dimensional), assume that it is higher dimensional. By assumption, $\Mm_\sigma$ is mixed, therefore $\As$ is spanned by o-lines (associated with real eigenvalues) and o-planes (associated with complex conjugate pairs). Let $\Vm \in \R^{\nx \times m}$ be a basis for $\As$ with $m > 2,$ where the $i$-th column of $\Vm$ is a real eigenvector of $\Mm_\sigma$ or, in case of a complex eigenvector pair $\vv, \vv^*$, the $i$-th and $(i+1)$-th columns are $\vv + \vv^*$ and $\imag\vv - \imag\vv^*$ respectively; these two columns correspond to an invariant plane of $\Mm_\sigma$. In the former case, we have
\begin{equation}
\Vm\Em_i = \vv,
\end{equation}
and in the latter
\begin{equation}
\Vm\Em_{i,i+1} = \begin{bmatrix}\vv + \vv^* & \imag\vv - \imag\vv^*\end{bmatrix},
\end{equation}
where $\Em_i$ is a row matrix with the $i$-th element being 1 and the rest zero, and $\Em_{i,i+1} \in \R^{2\times m}$ has the entries $(1,i)$ and $(2,i+1)$ equal to 1, the rest being zero. These are nothing but selection matrices.		

Since $\Qs_\sigma$ is composed of an intersection of sets of the form $\{\xv \in \R^\nx \mid \xv\tran\Qm_i\xv \sim 0\}$, where $\sim \in \{=,\geq,>\}$, by Prop.~\ref{prop:quadsubspace}, $\Vm\tran\Qm_i\Vm \approx \O$ for every such $\Qm_i$ determining $\Qs_\sigma$, where $\approx \in \{=, \succeq, \succ\}$, respectively. Since $\Am \approx \O \implies \Bm\tran\Am\Bm \approx \O$ for any non-singular $\Bm$, we can conclude that $\Vm\tran\Qm_i\Vm \approx \O$ implies $(\Vm\Em_i)\tran\Qm_i(\Vm\Em_i) \approx \O$  and $(\Vm\Em_{i,i+1})\tran\Qm_i(\Vm\Em_{i,i+1}) \approx \O$, which imply that the corresponding o-line or o-plane is also a subset of $\Qs_\sigma$.
\end{proof}

\subsection{Proofs of Theorem \ref{thm:topofixed} and Proposition \ref{prop:attractivity}}

In these proofs, if $f$ is not invertible in the pointwise sense, we treat its inverse in a set-based manner: $f^{-1} : \Ys \rightrightarrows \Xs$, $f^{-1}(y) = \{x \in \Xs \mid f(x) = y\}$. In addition, here we work on the real projective space $\P^{\nx-1},$ the space of all o-lines in $\R^\nx$. The real projective space is the quotient of $\R^n \setminus \{\O\}$ by the relation $\xv \sim \lambda\xv, \lambda \in \R \setminus \{0\}.$ Therefore, $\xv$ and $\lambda\xv$ are \emph{the same point} $\pv \in \P^{\nx-1}$. We denote the natural projection of a point in $\R^{\nx}$ onto $\P^{\nx-1}$ by $h : \R^\nx \setminus \{\O\} \to \P^{\nx-1}$.

\begin{lem}\label{lem:welldef}
Consider system \eqref{eq:samplemap} and assume $f(\xv) \neq \O$ for all $\xv \neq \O$. Then $g \coloneqq h \circ f \circ h^{-1}$ is a well-defined function. Moreover, if $f$ is continuous, then $g$ is also continuous.
\end{lem}
\begin{proof}
For any $\pv \in \P^{\nx-1},$ $h^{-1}(\pv)$ gives a whole o-line $\lv \subset \R^\nx$. From Prop.~\ref{prop:props}, it holds that $f(\lv) = \lv'$, where $\lv'$ is also an o-line. Hence $h(\lv') = \pv' \in \P^{\nx-1}$, so $g$ is well defined. %
Continuity is then given by the fact that $f$ is also continuous for o-lines, i.e., if $\lv$ is an o-line, $\lim_{\lv'\to\lv}f(\lv') = \lv$; hence $\lim_{\pv'\to\pv}g(\pv') = \pv$.
\end{proof}

\begin{customproof}[Proof of Theorem \ref{thm:topofixed}]
Every continuous map from the real projective space to itself has a fixed point if its dimension is even \cite[p.~109]{granas2003fixed}. From Lemma \ref{lem:welldef}, $g : \P^{\nx-1} \to \P^{\nx-1}$ is a well-defined continuous function; thus, $g$ has a fixed point if $\nx$ is \emph{odd}.

Now we need to show that, if $g$ has a fixed point, then $f$ has a fixed o-line. If $\pv$ is a fixed point of $g$, then take a point $\xv \in h^{-1}(\pv)$. Then, $h(f(\xv)) = g(h(\xv)) = \pv \therefore f(\xv) \in h^{-1}(\pv)$ Hence, there exists $\xv' \in h^{-1}(\pv)$ satisfying $\xv' = f(\xv)$, where $\xv' = \lambda\xv$, for some $\lambda$. Since $f$ is homogeneous as per Prop.~\ref{prop:props}, $\xv' = f(\xv)$ is true for any $\xv$ in the o-line containing it. Hence, this line is fixed by $f$, and the proof is complete.
\end{customproof}

\begin{customproof}[Proof of Prop.~\ref{prop:attractivity}]
From Lemma \ref{lem:welldef}, $g : \P^{\nx-1} \to \P^{\nx-1}$ is well defined. Let $\pv = h(\xv)$ for any $\xv \in \lv$. We want to show that there is a coordinate system for the tangent space of $g$ at $\pv$ such that the Jacobian of $g$ at $\pv$ is equal to $\frac{1}{\lambda}\Om_{\xv}\, \tran J_f (\xv) \Om_{\xv}$.

First, note that the real projective space is locally equal to the unit sphere, hence we can use the orthogonal subspace to a unitary $\xv$ within $\lv$ as the tangent subspace of $\pv$ embedded in $\R^\nx$. Denote it as $\Ts(\xv)$. Let $\dv$ be a unitary vector orthogonal to $\xv$. Any point in $\Ts(\xv)$ can be described as $\xv + a\dv$. To get the Jacobian of $g$, we apply $f$ to $\xv+h\dv$ and project the result back to $\Ts(\xv)$:
$f(\xv + h\dv) = f(\xv) + hJ_f(\xv)\dv + \Os(h^2) = \lambda\xv + hJ_f(\xv)\dv + \Os(h^2),$
whose projection back to $\Ts(\xv)$ is simply $\xv + h/\lambda \cdot J_f(\xv)\dv + \Os(h^2)$. Thus, the vector of variation of $g$ w.r.t.~$\dv$ embedded in $\R^{\nx}$ is
$$ \lim_{h\to 0} \frac{\xv + h/\lambda \cdot J_f(\xv)\dv + \Os(h^2) - \xv}{h} = \frac{1}{\lambda}J_f(\xv)\dv. $$
Now let $\dv_i$ be the $i$-th column of $\Om_{\xv}$. Every $\dv_i$ is unitary and orthogonal to $\xv$. Setting $\dv_1,\dv_2,...\dv_{\nx-1}$ as a coordinate system for the tangent space of $g$ at $\pv$, the component of the derivative of $g$ on $\dv_j$ from a variation in $\dv_i$ is $\dv_j\tran\frac{1}{\lambda}J_f(\xv)\dv_i$; putting in matrix form, we arrive at
$$ J_g(\pv) = \frac{1}{\lambda}\Om_{\xv}\, \tran J_f (\xv) \Om_{\xv}, $$
which implies local attractivity if Schur.
\end{customproof}

\subsection{Proof of Theorem \ref{thm:dimplane}}

We start by introducting the following lemma.

\begin{lem}\label{lem:fixedplane}
Let $\Nm \in \S^n$ be a nonsingular symmetric matrix. The following holds:
\begin{enumerate}
\item There is a plane through the origin $\Ps$ such that $\xv\in\Ps\setminus\{0\} \implies \xv\tran\Nm\xv > 0$ if and only if $\Nm$ has at least two positive eigenvalues.
\item There is a plane through the origin $\Ps$ such that $\xv\in\Ps \implies \xv\tran\Nm\xv = 0$ if and only if $n\geq4$, and $\Nm$ has at least two positive and two negative eigenvalues.
\end{enumerate}
\end{lem}
\begin{proof}
(1) This is a trivial consequence of Sylvester's law of inertia (see \cite[Chap.~XV.4]{lang2005algebra}).

(2) Based on Prop.~\ref{prop:quadsubspace}, this is equivalent to $\Vm\tran\Nm\Vm = \O,$ for some $\Vm \in \R^{n \times 2}.$ 

Proof of necessity: We assume that some full-rank $\Vm \in \R^{n \times 2}$ satisfies $\Vm\tran\Nm\Vm = \O$ and prove that $\Nm$ has at least two positive and two negative eigenvalues (thus $n\geq4$). Using Sylvester's law of inertia, we can write $\Nm = \Tm\tran\Sm\Tm$, where $\Sm$ is diagonal containing only 1 and $-1$ entries in the diagonal, and $\Tm$ is invertible. Thus, $\Vm\tran\Nm\Vm = \Wm_1\tran\Sm\Wm_1$, where $\Wm_1 = \Tm\Vm$ also has rank 2. Let $\Wm_2 \coloneqq \Sm\Wm_1.$ Since $\Sm$ is invertible, $\Wm_2$ has rank 2 as well. Because $\Wm_1\tran\Wm_2 = \O,$ the columns of $\Wm_1$ are orthogonal to the columns of $\Wm_2,$ hence $\Wm \coloneqq \begin{bmatrix}\Wm_1 & \Wm_2\end{bmatrix}$ has rank 4, which implies that $n \geq 4$. Pre-multiplying $\Wm_2 = \Sm\Wm_1$ by $\Sm$, we get $\Sm\Wm_2 = \Sm^2\Wm_1 = \Wm_1$ (note that $\Sm^2 = \I$). Thus, we can write
$$ \Sm\begin{bmatrix}\Wm_1 & \Wm_2\end{bmatrix} = \Sm\Wm = \begin{bmatrix}\Wm_2 & \Wm_1\end{bmatrix} = \Wm\Pm, $$
where $\Pm \coloneqq \begin{bsmallmatrix}\O & \I \\ \I & \O\end{bsmallmatrix}$ is a permutation matrix. This matrix has two eigenvalues in 1 and two eigenvalues in -1. 
Now, take one pair $(\lambda, \xv)$ such that $\Pm\xv = \lambda\xv$. Then, $\Sm\Wm\xv = \Wm\Pm\xv = \lambda\Wm\xv,$ so $\lambda$ is also an eigenvalue of $\Sm$. Thus, $\Sm$ has at least two eigenvalues equal to 1 and two equal to $-1$.

Proof of sufficiency: now we start with a nonsingular matrix $\Nm$ with two positive and two negative eigenvalues, and then construct $\Vm \in \R^{n \times 2}$ such that $\Vm\tran\Nm\Vm = \O$. Take the Sylvester matrix $\Sm$ of $\Nm$ and select 4 rows and columns such that the corresponding submatrix has exactly two values of 1 and two of $-1$. Denote by this submatrix $\Sm_4$. We have that that there exists $\Wm_4 \in \R^{4 \times 4}$ such that $\Wm_4^{-1}\Sm_4\Wm_4 = \Pm,$ where $\Pm$ the same permutation matrix as in the proof of necessity, so $\Wm_4$ can be determined by the eigendecomposition of $\Pm$. Complete $\Wm \in \R^{n \times 4}$ from $\Wm_4$ by padding the remaining rows with zeros, and denote the first two columns of $\Wm$ by $\Wm_1$. Using the same arguments as in the proof of necessity, the matrix $\Vm = \Tm^{-1}\Wm_1$ satisfies $\Vm\tran\Nm\Vm = \O$. 
\end{proof}	

\begin{customproof}[Proof of Theorem \ref{thm:dimplane}]
Case (i) is trivial, since $\underline\tau = \bar\tau$ implies periodic sampling, so the whole $\R^{\nx}$ is fixed and isochronous. 

For case (ii), suppose $\nx < 3$; then every isochronous set is composed by quadratic sets of the form $\cap_i\{\xv \in \R^\nx \mid \xv\tran\Qm_i\xv\mathrel{>\!\!(\geq)}0\}$, but by Lemma \ref{lem:fixedplane} (1) every such $\Qm_i$ must have two positive eigenvalues; thus, every $\Qm_i \succ \O$ which implies that $\Qs_\tau = \R^\nx$ for all $k$, hence every state samples at every possible inter-sample time. Because a state can only sample at one inter-sample time, this implies $\tau$ is unique, hence $\underline\tau = \bar\tau$, which is a contradiction.

For case (iii), in the CETC case every set $\Qs_\tau$ is an interesecion of sets including the set $\{\xv \in \R^\nx \mid \xv\tran\Qm_\tau\xv=0\}$. Hence, by Lemma \ref{lem:fixedplane} (2), if $\nx < 4$ then no plane can belong to $\Qs_\tau$; thus $\nx \geq 4$.
\end{customproof}

\subsection{Proof of Proposition \ref{prop:unstable}}

First we introduce the following Lemma.

\begin{lem}
\label{lem:conestable} Let $\xiv(k+1) = \Mm\xiv(k)$ be a linear autonomous system, $\Mm$ mixed, and let $\vv_1, \vv_2, ..., \vv_n$ be the unitary eigenvectors of $\Mm$ ordered from largest-in-magnitude corresponding eigenvalue to smallest. Denote by $\As$ any linear invariant of $\Mm$ containing $\vv_1$. Then, for every initial state $\xiv(0) = a_1\vv_1 + \cdots = a_n\vv_n$ where $a_1 \neq 0$, it holds that
$$ \lim_{k\to\infty}\frac{\xiv(k)}{|\xiv(k)|} \in \As. $$
\end{lem}
\begin{proof}
This is consequence of the proof of \cite[Lemma 3]{gleizer2021computing} when $a_1 \neq 0$. We omit the details here due to space limitations.
\end{proof}

Lemma \ref{lem:conestable} paraphrases the known fact that almost every trajectory of a linear system converges to its dominant mode.

\begin{customproof}[Proof of Proposition \ref{prop:unstable}]
{\bf Item (I):} For contradiction, assume that $\sigma^\omega$ is stable, and let $m \coloneqq |\sigma|.$ First, we check $\xv_0 \in \Qs_\sigma$. In this case, the samples $\xv_i$ evolve according to $\xv_{i+m} = \Mm_\sigma\xv_i$. Let $\xv_0 = a_1\vv_1 + \cdots = a_n\vv_n$ where $\vv_j$ are the eigenvectors of $\Mm_\sigma$ ordered as in Lemma \ref{lem:conestable}. For any $\xv_0 \in \R^n$ almost all points in its neighborhood satisfy $a_1 \neq 0$. Hence, by Lemma \ref{lem:conestable}, $\lim_{i\to\infty}\frac{\xv_{mi}}{|\xv_{mi}|} \in \As$, but $\As \nsubseteq \cl(\Qs_\sigma)$. Thus, $\{\xv_{mi}\}$ escapes $\Qs_\sigma$ at a some finite $i$, hence $\{y(\xv_{mi})\} \neq \sigma^\omega$. This contradicts the assumption that $\sigma^\omega$ is stable.

We now see that the set of states $\xv$ such that $\Bs^\omega_{\xv}(\Ss) = \alpha\sigma^\omega$, $|\alpha| < \infty$ is measure zero. This is done by induction on the length of $\alpha$. Let $\Xs_m$ be the set of states whose behavior is if $\alpha\sigma^\omega$ with $|\alpha| = m.$ If $m=0$, we have already seen that $\Xs_0$ is a linear invariant of $\Mm_\sigma$; because this linear subspace does not contain $\vv_1$, it has zero measure. Now assume $\Xs_m$ has zero measure. The set $\Xs_{m+1}$ is the pre-image of $\Xs_m$, hence, $\Xs_{m+1} \subseteq \cup_{k=1}^{\bar{k}}\Mm_k^{-1}\Xs_m$. Because $\Mm_k$ is nonsingular and the union is finite, $\Xs_{m+1}$ is also measure zero. This concludes the proof that $\sigma^\omega$ is unstable.

{\bf Item (II):}
First, note that $c$ must be a multiple of $|\sigma|$. Let $\alpha$ be any $l$-long subsequence of $\sigma^\omega$. We have already seen that for almost every $\xv \in \Qs_\alpha$ there exists a finite $k$ such that the solution to \eqref{eq:samplemap} $\xiv_{\xv}(ck-1) \notin \Qs_\alpha.$ 
To prove absolutely instability, it suffices to check the behavior from any such $\xiv_{\xv}(ck-1)$ does not contain $\sigma^L$ for some $L$ large enough. We show that it is true with $L=c/|\sigma|$.

Let $\Cs$ be the simple-cycle SCC formed by $\{x_1, x_2, ..., x_c\}$. Since $\alpha \in \Cs$, w.l.o.g., let $x_1=\alpha,$ which is related to $\xv$. Let $\xv' \coloneqq \xiv_{\xv}(ck)$ and take $x'$ as the unique state in $\Ss_l$ related to $\xv'$, respectively. We first show that $x' \notin \Cs$: since there is a run from $\xv$ to $\xv'$ with length $ck$, there must be a run segment of length $ck$ from $x_1$ to $x'$. Because $\Cs$ is strongly connected, if $x' \in \Cs$, the only path would be $(x_1...x_c)^kx_1,$ hence $x' = x_1$. But this is a contradiction because $x' \neq \alpha$ since $\xiv_{\xv}(ck-1) \notin \Qs_\alpha.$ Thus, $x' \notin \Cs$.

Now, there is no path in the abstraction connecting $x'$ to $\Cs$ (otherwise $\Cs$ would not be a simple cycle). Therefore, because $\Ss_l \succeq \Ss$, it is trivial to see that there is also no path from $\xv'$ back to $\Qs_\alpha$, i.e., $\xiv_{\xv'}(k) \notin \Qs_\alpha, \forall k \in \N$. Thus, $\Bs_{\xv'}^\omega(\Ss)$ does not contain $\sigma^L$ as a subsequence, concluding the proof.
\end{customproof}

\end{document}